\documentclass[a4paper,12pt]{article}
\usepackage{amsmath,amssymb,bm}
\usepackage{enumerate}
\usepackage{mathrsfs}
\usepackage{amsthm}

 \newtheorem{thm}{Theorem}
 \newtheorem{cor}{Corollary}
 \newtheorem{lem}{Lemma}
 \newtheorem{prop}{Proposition}
 
 \newtheorem{step}{Step}
 \theoremstyle{definition}
 
 \theoremstyle{remark}

\numberwithin{equation}{section}

\title{The Hamilton-Jacobi-Bellman Equation in Economic Dynamics with a Non-Smooth Fiscal Policy}
\author{Yuhki Hosoya\thanks{E-mail: hosoya(at)tamacc.chuo-u.ac.jp, ORCID ID: 0000-0002-8581-4518}\\Faculty of Economics, Chuo University\thanks{742-1 Higashinakano, Hachioji-shi, Tokyo 192-0393, Japan.}}
\date{\today}
\begin{document}
\maketitle

\begin{abstract}
We consider a class of economic growth models that includes the classical Ramsey--Cass--Koopmans capital accumulation model and verify that, under several assumptions, the value function of the model is the unique viscosity solution to the Hamilton--Jacobi--Bellman equation. Moreover, we discuss a solution method for these models using differential inclusion, where the subdifferential of the value function plays an important role. Next, we present an assumption under which the value function is a classical solution to the Hamilton--Jacobi--Bellman equation, and show that many economic models satisfy this assumption. In particular, our result still holds in an economic growth model in which the government takes a non-smooth Keynesian policy rule.

\vspace{12pt}

\vspace{12pt}
\end{abstract}

\section{Introduction}
Since the 1970s, dynamic theory in macroeconomics has commonly been modelled using the calculus of variations or optimal control theory. The most important reason for the emergence of such a culture was the observation in the developed economies of the 1970s of phenomena that could not be explained by the traditional theory of fiscal and monetary policy based on the work of Phillips \cite{PH}. Lucas \cite{LU} explained that this stemmed from a flaw in the previous theory, and this was considered the most persuasive view in economics. The best method for removing the flaw that Lucas criticized was seen to be the use of optimal control theory in the model, so this style became popular. In particular, the so-called Ramsey--Cass--Koopmans (RCK) model, constructed by Ramsey \cite{RA} and modified by Cass \cite{CA} and Koopmans \cite{KO} for ease of use, is the most fundamental of these models, and is frequently used in textbooks of advanced macroeconomics: see, for example, Acemoglu \cite{AC}, Barro and Sara-i-Martin \cite{BS1}, Blanchard and Fischer \cite{BF}, or Romer \cite{ROM}.

Unfortunately, the models treated in modern theories of macroeconomic dynamics are far more complex than the RCK model and, in most cases, cannot be solved directly.\footnote{In some cases, there are exceptions that can be solved using classical methods including Pontryagin's maximum principle. See, for example, Huong \cite{HU}. In the case of linear constraints, the classical arguments concerning an explicit solution can also be found in Chapter 4 of Barro and Sara-i-Martin \cite{BS1}. Hosoya and Kuwata \cite{HK} extend this result to include overtaking optimal cases.} Therefore, analysis using the Hamilton--Jacobi--Bellman (HJB) equation is commonly applied. Achdou et al. \cite{ABLLM} is a typical example of such analysis. In fact, optimal control problems in economics have a different form than that classically treated by Lions \cite{LI} and Crandall and Lions \cite{CL2}. For such a class of problems, Barles \cite{BA} pointed out that the HJB equation may behave in an unusual way. In economics, however, a `proof' of the relationship between the value function of the problem and the HJB equation was provided early on, as dealt with in Marrialis and Brock \cite{MB}. Therefore, it was believed that the problem noted by Barles did not occur in economics.

Hosoya \cite{HOB} identified a gap of this `proof' of the relationship between the value function and the HJB equation in economics. This gap is so serious that we can construct an economic model in which the HJB equation has no relationship with the value function. Therefore, the problem specified by Barles is important even in economics, making it necessary to find conditions under which the problem vanishes. Hosoya \cite{HOA} addressed this problem and showed that, under certain conditions, the value function is the unique classical solution to the HJB equation in some functional space.

However, Hosoya \cite{HOA} assumed that the functions being treated in the model are smooth. In economic models, it may be possible to assume that the utility function (which expresses the consumer's preference, and the technology function, which describes the production process) are smooth, but we cannot be sure that the {\bf fiscal policy rule} owned by the government is smooth. Indeed, in this paper, we construct an example of an economic model in which the government uses a Keynesian fiscal policy rule that is non-smooth (see Subsection 4.2). Therefore, the assumption of the smoothness is undesirable for functions in the control problem of economic models that take fiscal policy into account.

This is exactly the problem we focus on in this paper. In other words, our purpose is to derive a relationship between the HJB equation and the value function without the differentiability assumptions of the technology function considered by Hosoya \cite{HOA}. In this case, because the model is not differentiable, it is natural that the solution of the HJB equation is also not differentiable. Therefore, unlike Hosoya \cite{HOA}, we mainly use the notion of viscosity solutions. In Hosoya \cite{HOA}, it was proved that the value function is a viscosity solution to the HJB equation under several assumptions. In this paper, we show a similar result under weaker conditions (Theorem 1).

Furthermore, this paper attempts to determine a condition under which there is no other viscosity solution to the HJB equation than the value function. This problem is quite non-trivial. As treated in Bardi and Capuzzo-Dolcetta \cite{BCD}, the viscosity solutions to the HJB equation are usually discussed in the class of bounded, uniformly continuous functions. In this case, for any pair of a viscosity subsolution and a viscosity supersolution, the former is less than the latter, from which the uniqueness of the viscosity solution can easily be derived. However, this argument is possible because boundedness can be naturally introduced into the objective function. The most typical objective function in economic control problems is $\int_0^{\infty}e^{-\rho t}\log c(t)dt$, where $c(t)$ is only assumed to be non-negative and locally integrable. In this case, the value function is usually unbounded in the first place. Thus, the above class of functions cannot be used as the set of candidates for solutions. Fortunately, in most economic models, the value function is concave. However, when considering the class of concave functions, there are often multiple solutions to the HJB equation (see the beginning of Subsection 3.5).

Therefore, as the set of candidates for solutions, we consider the space of increasing and concave functions that satisfy a sort of growth condition. All functions included in this class have positive subderivatives. Therefore, we can evaluate these functions using differential inclusion defined in terms of subderivatives. For problems satisfying the assumptions treated in Theorem 1 and one additional assumption, we show that the value function is the unique viscosity solution to the HJB equation in this class (Theorem 2).

The differential inclusion that we consider is interesting in its own right. Using this differential inclusion and the value function, we can obtain the solution to the original optimal control problem (Corollary 1). Furthermore, under an additional assumption, the value function becomes differentiable, whereupon the differential inclusion becomes a usual differential equation. Hence, we can prove that a solution to the optimal control problem can be taken continuously (Corollary 2). These are the main results in this paper.

Finally, we confirm that our results are correctly applicable to economic models. We do not assume differentiability in our model, but it is, of course, possible to apply our results to models that do assume differentiability. Therefore, we first consider the usual RCK model and check that our assumptions are very weak. Next, we introduce a non-smooth Keynesian fiscal policy for such models. Even in such a case, our assumptions hold naturally, and thus the results mentioned above are all applicable. Surprisingly, even though our constructed fiscal policy function is not differentiable, if capital depreciation is absent, then the value function is differentiable. However, this is not always the case in models with positive capital depreciation.

The remaining of this paper is organized as follows. First, Section 2 discusses the model and its underlying assumptions, and presents rigorous definitions of terms such as the value function. Section 3 derives the main results. Section 4 confirms the applicability of our results to economic models. Section 5 discusses the relationship with related studies. Section 6 presents the conclusion. Several proofs of lemmas is contained in the appendix.

\section{Model and Definitions}
\subsection{Model}
The model that we discuss in this paper is the following.\footnote{The statement `$\int_0^{\infty}e^{-\rho t}u(c(t),k(t))dt$ can be defined' admits that the value of this integral may be $\pm \infty$.}
\begin{align}
\max~~~~~&~\int_0^{\infty}e^{-\rho t}u(c(t),k(t))dt \nonumber \\
\mbox{subject to. }&~c(\cdot)\in W,\nonumber \\
&~k(t)\mbox{ is locally absolutely continuous},\nonumber \\
&~k(t)\ge 0,\ c(t)\ge 0,\label{MODEL}\\
&~\int_0^{\infty}e^{-\rho t}u(c(t),k(t))dt\mbox{ can be defined},\nonumber \\
&~\dot{k}(t)=F(k(t),c(t))\mbox{ a.e.},\nonumber \\
&~k(0)=\bar{k},\nonumber
\end{align}
where $W$ denotes some set of functions. The symbol $c(t)$ means the amount of consumption, and $k(t)$ denotes the amount of capital stock. If the instantaneous utility function $u(c,k)$ is independent of $k$ and the technology function $F(k,c)$ is $f(k)-dk-c$, then problem (\ref{MODEL}) coincides with the traditional RCK model.\footnote{The constant $d\ge 0$ denotes the capital depreciation rate.} In the context of the RCK model, the requirement that $\mathbb{R}_{++}=f'(\mathbb{R}_{++})$ (called the {\bf Inada condition}) is sometimes used. Later, we treat such a case (see Subsection 4.1).

We implicitly consider the following model: there exists a {\bf fiscal policy rule} $G(c,k)$ that is exogenously determined, and the government expenditure $g(t)$ is determined by the equation $g(t)=G(c(t),k(t))$. Moreover, $F(k,c)=f(k)-dk-c-G(c,k)$, and $u(c,k)=v(c,G(c,k))$, where $v$ is some function. In economic models, $f$ and $v$ are usually assumed to be differentiable, but $G$ may not be differentiable. Hence, $F$ and $u$ may not be differentiable. However, the non-differentiability of $u$ significantly increases the difficulty of the analysis, and thus, we later make a compromise and assume that $u$ is differentiable.\footnote{This can be naturally justified if $v$ is independent of $g$. Later, we discuss the meaning of such an assumption.}

Let $W_1$ be the set of all functions $c:\mathbb{R}_+\to\mathbb{R}_+$ that are locally integrable, and $W_2$ be the set of all functions $c:\mathbb{R}_+\to\mathbb{R}_+$ that are measurable and locally bounded. Note that $W_2\subset W_1$.\footnote{In this paper, we use the following notation. First, we define $\mathbb{R}^n_+=\{x\in \mathbb{R}^n|x_i\ge 0\mbox{ for all }i\}$ and $\mathbb{R}^n_{++}=\{x\in \mathbb{R}^n|x_i>0\mbox{ for all }i\}$. If $n=1$, we simply write these sets as $\mathbb{R}_+$ and $\mathbb{R}_{++}$, respectively.}

\vspace{12pt}
\noindent
\textbf{Assumption 1}. $\rho>0$, and $W$ is either $W_1$ or $W_2$.

\vspace{12pt}
\noindent
\textbf{Assumption 2}. The instantaneous utility function $u:\mathbb{R}^2_+\to \mathbb{R}\cup \{-\infty\}$ is a continuous and concave function on $\mathbb{R}^2_+$. Moreover, $u(c,k)$ is nondecreasing on $\mathbb{R}^2_+$, and increasing in $c$ and continuously differentiable on $\mathbb{R}^2_{++}$.\footnote{Because $u$ is increasing in $c$ on $\mathbb{R}^2_{++}$, if $(c,k)\in \mathbb{R}^2_{++}$, then $u(c,k)>u(2^{-1}c,k)\ge -\infty$. Therefore, $u(c,k)\in \mathbb{R}$ for every $(c,k)\in \mathbb{R}^2_{++}$.} Furthermore, there exists $c>0$ such that $u(c,0)>-\infty$.

\vspace{12pt}
\noindent
\textbf{Assumption 3}. The technology function $F:\mathbb{R}^2_+\to \mathbb{R}$ is a continuous and concave function that satisfies $F(0,0)=0$. Moreover, $F$ is decreasing in $c$, and there exist $d_1>0$ and an increasing function $\delta_2:\mathbb{R}_+\to\mathbb{R}_+$ such that $\delta_2(0)=0$ and $F(k,c)>-d_1k-\delta_2(c)$ for every $(k,c)$ such that $k>0$ and $c\ge 0$. If $W=W_1$, then there exists $d_2\ge 0$ such that $\delta_2(c)=d_2c$ for all $c\ge 0$. Furthermore, for every $c\ge 0$, there exists $k>0$ such that $F(k,c)>F(0,c)$.\footnote{The last assumption means that, for every $\bar{c}>0$, there exists $k^*>0$ such that for every $c\in [0,\bar{c}]$, $k\mapsto F(k,c)$ is increasing on $[0,k^*]$. In fact, we can choose
\[k^*=\min_{c\in [0,\bar{c}]}\min\arg\max\{F(k,c)|0\le k\le 1\}.\]}

\vspace{12pt}
\noindent
\textbf{Assumption 4}. The function $\frac{\partial u}{\partial c}(c,k)$ is decreasing in $c$ on $\mathbb{R}^2_{++}$. Moreover, $\lim_{c\to 0}\frac{\partial u}{\partial c}(c,k)=+\infty$ and $\lim_{c\to \infty}\frac{\partial u}{\partial c}(c,k)=0$ for every $k>0$. Furthermore, for every $k>0$, $\limsup_{c\downarrow 0}\frac{\partial u}{\partial k}(c,k)<+\infty$.

\subsection{Admissibility and the Value Function}

We say that a pair of real-valued functions $(k(t),c(t))$ defined on $\mathbb{R}_+$ is \textbf{admissible} if $k(t)$ is absolutely continuous on every compact interval, $c(t)\in W$, $k(t)\ge 0, c(t)\ge 0$, $\int_0^{\infty}e^{-\rho t}u(c(t),k(t))dt$ can be defined, and
\begin{equation}
\dot{k}(t)=F(k(t),c(t))\ \mbox{a.e.}.\label{TC}
\end{equation}
Note that, if $k(t)$ is absolutely continuous on every compact interval, then it is differentiable almost everywhere and $\int_a^b\dot{k}(t)dt=k(b)-k(a)$ for all $a,b$ with $0\le a<b$.

Let $A_{\bar{k}}$ denote the set of all admissible pairs such that $k(0)=\bar{k}$. Using the notation $A_{\bar{k}}$, we can simplify model (\ref{MODEL}) as follows:
\begin{align*}
\max~~~~~&~\int_0^{\infty}e^{-\rho t}u(c(t),k(t))dt\\
\mbox{subject to. }&~(k(t),c(t))\in A_{\bar{k}}.
\end{align*}
For each $\bar{k}>0$, let
\[\bar{V}(\bar{k})=\sup\left\{\left.\int_0^{\infty}e^{-\rho t}u(c(t),k(t))dt\right| (k(t),c(t))\in A_{\bar{k}}\right\}.\]
We call this function $\bar{V}$ the \textbf{value function} of problem (\ref{MODEL}).\footnote{We can easily show that $A_{\bar{k}}$ is nonempty for all $\bar{k}\ge 0$. Actually, $(k(t),c(t))\equiv (0,0)\in A_0$, and if $\bar{k}>0$, this result is proved in the proof of Proposition 2 of Hosoya \cite{HOA}. Therefore, we can define $\bar{V}$ on $\mathbb{R}_+$. However, in this paper, we consider the domain of $\bar{V}$ to be $\mathbb{R}_{++}$ for some technical reasons; thus, $\bar{V}(0)$ is not treated throughout this paper, although it can be defined.}

We call a pair $(k^*(t),c^*(t))\in A_{\bar{k}}$ a \textbf{solution} if and only if the following two requirements hold. First,
\[\int_0^{\infty}e^{-\rho t}u(c^*(t),k^*(t))dt\in \mathbb{R}.\]
Second, for every pair $(k(t),c(t))\in A_{\bar{k}}$,
\[\int_0^{\infty}e^{-\rho t}u(c^*(t),k^*(t))dt\ge \int_0^{\infty}e^{-\rho t}u(c(t),k(t))dt.\]
These two requirements can be summarized by the following formula:
\[\int_0^{\infty}e^{-\rho t}u(c^*(t),k^*(t))dt=\bar{V}(\bar{k})\in \mathbb{R}.\]
For all $(k(t),c(t))\in A_{\bar{k}}$ such that $\int_0^{\infty}e^{-\rho t}u(c(t),k(t))dt\in \mathbb{R}$ and $T>0$, if $k(T)>0$, then
\begin{align*}
\int_0^Te^{-\rho t}u(c(t),k(t))dt=&~\int_0^{\infty}e^{-\rho t}u(c(t),k(t))dt-\int_T^{\infty}e^{-\rho t}u(c(t),k(t))dt\\
\ge&~\int_0^{\infty}e^{-\rho t}u(c(t),k(t))dt-e^{-\rho T}\bar{V}(k(T)).
\end{align*}
In particular, if $\int_0^{\infty}e^{-\rho t}u(c(t),k(t))dt>M$, then for all $T>0$ such that $k(T)>0$,
\[\int_0^Te^{-\rho t}u(c(t),k(t))dt>M-e^{-\rho T}\bar{V}(k(T)).\]
Therefore, if $\bar{V}(\bar{k})$ is finite, then for every $\varepsilon>0$, there exists a pair $(k(t),c(t))\in A_{\bar{k}}$ such that either $k(T)=0$ or
\[\int_0^Te^{-\rho t}u(c(t),k(t))dt>\bar{V}(\bar{k})-e^{-\rho T}\bar{V}(k(T))-\varepsilon\]
for every $T>0$.

\subsection{HJB Equation}
The HJB equation is given as follows.
\begin{equation}
\sup_{c\ge 0}\{F(k,c)V'(k)+u(c,k)\}-\rho V(k)=0.\label{HJB}
\end{equation}
A function $V:\mathbb{R}_{++}\to \mathbb{R}\cup \{-\infty\}$ is called a \textbf{classical solution} to the HJB equation if and only if $V$ is continuously differentiable and equation (\ref{HJB}) holds for every $k>0$.

In many models of the dynamic control problem, there exists no classical solution to the HJB equation. Hence, we should extend the notion of the solution. First, a function $V:\mathbb{R}_{++}\to \mathbb{R}$ is called a \textbf{viscosity subsolution} to (\ref{HJB}) if and only if it is upper semi-continuous, and for every $k>0$ and every continuously differentiable function $\varphi$ defined on a neighbourhood of $k$ such that $\varphi(k)=V(k)$ and $\varphi(k')\le V(k')$ whenever $k'$ is in the domain of $\varphi$,\footnote{Usually, the condition for $\varphi$ is that $\varphi(k')\ge V(k')$, because this equation is considered for some minimization problem. In this paper, however, we treat a maximization problem (\ref{MODEL}), and thus the inequality is reversed.}
\[\sup_{c\ge 0}\{F(k,c)\varphi'(k)+u(c,k)\}-\rho V(k)\le 0.\]
Second, a function $V:\mathbb{R}_{++}\to \mathbb{R}$ is called a \textbf{viscosity supersolution} to (\ref{HJB}) if and only if it is lower semi-continuous, and for every $k>0$ and every continuously differentiable function $\varphi$ defined on a neighbourhood of $k$ such that $\varphi(k)=V(k)$ and $\varphi(k')\ge V(k')$ whenever $k'$ is in the domain of $\varphi$,
\[\sup_{c\ge 0}\{F(k,c)\varphi'(k)+u(c,k)\}-\rho V(k)\ge 0.\]
If a continuous function $V:\mathbb{R}_{++}\to \mathbb{R}$ is both a viscosity sub- and supersolution to (\ref{HJB}), then $V$ is called a \textbf{viscosity solution} to (\ref{HJB}).

Suppose that $V$ is a viscosity solution to the HJB equation and is differentiable at $k>0$. Then, it is known that
\[\sup_{c\ge 0}\{F(k,c)V'(k)+u(c,k)\}-\rho V(k)=0.\]
See Proposition 1.9 of Ch.2 of Bardi and Capuzzo-Dolcetta \cite{BCD}.

\vspace{12pt}
\noindent
{\bf Note}: In many studies, it is proved that for any viscosity subsolution $v_1$ and viscosity supersolution $v_2$, $v_1\le v_2$. If so, it is trivial that if there exists a viscosity solution, then such a solution is unique. However, in (\ref{MODEL}), such a result cannot be proved. For a simple counter-example, see the first paragraph of Subsection 3.5.

\subsection{Subdifferentials and Left- and Right-Derivatives}
In this paper, we heavily use subdifferential calculus. We introduce the notion of the subdifferential and several results. For the proofs of these results, see textbooks on convex analysis, such as Rockafeller \cite{ROC}.

Suppose that a function $G:U\to \mathbb{R}$ is concave, $U\subset \mathbb{R}^n$ is convex, and the interior $V$ of $U$ is nonempty. Choose any $x\in V$. We define
\[\partial G(x)=\{p\in \mathbb{R}^n|G(y)-G(x)\le p\cdot (y-x)\mbox{ for all }y\in U\}.\]
Then, we can show that $\partial G(x)$ is nonempty. The set-valued mapping $\partial G$ is called the \textbf{subdifferential} of $G$.\footnote{Formally, the subdifferential is defined for not concave but \textbf{convex} functions, and thus the inequality in the definition is reversed. In this view, the name `subdifferential' may not be appropriate, and `superdifferential' may be more suitable. However, in the literature of economics, these two notions are not distinguished, and thus our $\partial G$ is traditionally called the `subdifferential'.} It is known that $G$ is differentiable at $x$ if and only if $\partial G(x)$ is a singleton, and if so, $\partial G(x)=\{\nabla G(x)\}$.

If $n=1$, then define the left- and right-derivatives $D_-G(x),\ D_+G(x)$ as
\[D_-G(x)=\lim_{y\uparrow x}\frac{G(y)-G(x)}{y-x},\ D_+G(x)=\lim_{y\downarrow x}\frac{G(y)-G(x)}{y-x}.\]
Note that, if $G$ is concave, then $\frac{G(y)-G(x)}{y-x}$ is nonincreasing in $y$, and thus
\[D_-G(x)=\inf_{t>0}\frac{G(x-t)-G(x)}{-t},\ D_+G(x)=\sup_{t>0}\frac{G(x+t)-G(x)}{t},\]
which implies that both $D_-G(x),D_+G(x)$ are defined and real numbers. It is known that $\partial G(x)=[D_+G(x),D_-G(x)]$.

Recall that, under Assumption 3, our $F$ is concave. In this case, the functions $c\mapsto F(k,c)$ and $k\mapsto F(k,c)$ are also concave, and thus the `partial' subdifferential can be considered. Let
\[\partial_kF(k,c)=\{p|F(k',c)-F(k,c)\le p(k'-k)\mbox{ for all }k'\ge 0\},\]
\[\partial_cF(k,c)=\{p|F(k,c')-F(k,c)\le p(c'-c)\mbox{ for all }c'\ge 0\}.\]
The partial left- and right-derivatives can be defined in the same manner. For example, 
\[D_{k,+}F(k,c)=\sup_{t>0}\frac{F(k+t,c)-F(k,c)}{t},\]
\[D_{k,-}F(k,c)=\inf_{t>0}\frac{F(k-t,c)-F(k,c)}{-t}.\]
Now, suppose that $f$ is concave and continuous, and $k_1<k_2$. Define $r=\frac{f(k_2)-f(k_1)}{k_2-k_1}$. Let $p=D_+f(k_1)$ and $q=D_-f(k_2)$. If $p=r$ or $q=r$, then $f'(k)\equiv r$ for all $k\in ]k_1,k_2[$. Suppose that $p>r>q$. Define $g(k)=f(k)-rk$. Then, $D_+g(k_1)>0$ and $D_-g(k_2)<0$, and thus, there exists $k\in ]k_1,k_2[$ such that $g(k)=\max_{k'\in [k_1,k_2]}g(k')$. By the definition of the subdifferential, $0\in \partial g(k)$, and thus $r\in \partial f(k)$. In conclusion, we obtain the following {\bf mean value theorem}: if $f$ is concave and $k_1<k_2$, then there exist $k\in ]k_1,k_2[$ and $r\in \partial f(k)$ such that $f(k_2)-f(k_1)=r(k_2-k_1)$.

Suppose that $G:I\to \mathbb{R}$ is a concave function, where $I$ is some open interval in $\mathbb{R}$. We prove the upper hemi-continuity of $\partial G(x)$.\footnote{Let $X,Y$ be some topological spaces and $F:X\twoheadrightarrow Y$ be some set-valued mapping. It is said that $F$ is {\bf upper hemi-continuous} at $x$ if for any open set $V\subset Y$ that includes $F(x)$, there exists an open neighborhood $U$ of $x$ such that $F(y)\subset V$ for all $y\in U$.} Suppose not. Then, there exist $\varepsilon>0$ and a sequence $(x_m)$ such that $x_m\to x$ as $m\to \infty$ and either $D_+G(x_m)\le D_+G(x)-2\varepsilon$ or $D_-G(x_m)\ge D_-G(x)+2\varepsilon$ for each $m$. We only treat the case in which the former holds for all $m$, because the other cases can be treated symmetrically. By the definition of the right-derivative, there exists $h>0$ such that
\[G(x+h)-G(x)\ge h[D_+G(x)-\varepsilon].\]
Therefore, for any sufficiently large $m$,
\[G(x_m+h)-G(x_m)>h[D_+G(x)-2\varepsilon]\ge hD_+G(x_m),\]
which is a contradiction. 

We use these facts in the proof of Proposition \ref{Prop4}.

\subsection{Pure Accumulation Path}
Consider the following differential equation:
\begin{equation}
\dot{k}(t)=F(k(t),0),\ k(0)=\bar{k}.\label{PAP}
\end{equation}
Let $k^+(t,\bar{k})$ denote the solution to the above equation defined on $\mathbb{R}_+$. This function $k^+(t,\bar{k})$ is called the \textbf{pure accumulation path}. The following lemma has been proved in Hosoya \cite{HOA}.\footnote{Note that, under Assumption 3, there exists $k>0$ such that $F(k,0)>0$.}

\begin{lem}\label{Lemma1}
Under Assumption 3, the pure accumulation path $k^+$ is uniquely defined on the set $\mathbb{R}_+\times\mathbb{R}_{++}$. Moreover, if $F(k,0)>0$ and $\gamma\in \partial_kF(k,0)$, then $k^+(t,\bar{k})\ge \min\{k,\bar{k}\}$ and
\begin{equation}
k^+(t,\bar{k})\le \begin{cases}
e^{\gamma t}\left(\bar{k}+(e^{-\gamma t}-1)\frac{\gamma k-F(k,0)}{\gamma}\right) & \mbox{if }\gamma\neq 0,\\
\bar{k}+tF(k,0) & \mbox{if }\gamma=0,
\end{cases}\label{SGC}
\end{equation}
for all $t\ge 0$.
\end{lem}

Let $V:\mathbb{R}_{++}\to \mathbb{R}$. The following requirement of $V$ is called the \textbf{growth condition}.
\begin{equation}
\lim_{T\to\infty}e^{-\rho T}V(k^+(T,\bar{k}))=0\mbox{ for all }\bar{k}>0.\label{GC}
\end{equation}
Define $\mathscr{V}$ as the space of all functions $V:\mathbb{R}_{++}\to \mathbb{R}$ that is increasing and concave, and satisfies the growth condition. Note that, by concavity, every $V\in \mathscr{V}$ is locally Lipschitz.

The following lemma shows that (\ref{GC}) is not strong.

\begin{lem}\label{Lemma2}
Suppose that there exist $k>0$ and $\gamma\in \mathbb{R}$ such that $\gamma\in \partial_kF(k,0)$ and $\gamma<\rho$. Then, $\mathscr{V}$ coincides with the set of all increasing and concave real-valued functions on $\mathbb{R}_{++}$.
\end{lem}

\begin{proof}
Without loss of generality, we can assume that $F(k,0)>0$. Therefore, by (\ref{SGC}), there exist $A,B,C\in \mathbb{R}$ such that $k^+(t,\bar{k})\le Ae^{\gamma t}+Bt+C$. Choose any increasing and concave function $V:\mathbb{R}_{++}\to \mathbb{R}$ and $p\in \partial V(\bar{k})$. Then,
\begin{align*}
-\infty<&~V(\inf_{t\ge 0}k^+(t,\bar{k}))\le V(k^+(T,\bar{k}))\\
\le&~V(\bar{k})+p(k^+(T,\bar{k})-\bar{k}),\\
\le&~V(\bar{k})+p(Ae^{\gamma T}+BT+C-\bar{k}),
\end{align*}
which implies that $V$ automatically satisfies (\ref{GC}). Thus, the requirement (\ref{GC}) vanishes and $\mathscr{V}$ coincides with the set of all increasing and concave functions on $\mathbb{R}_{++}$. This completes the proof.
\end{proof}

\section{Results}
In this section, we analyze the HJB equation for a characterization of the value function. We prohibit ourselves from making assumptions on the solution explicitly. Assumptions must be made for properties of primitives $\rho, u, F$ in our model (\ref{MODEL}), and, for example, the assumption for the existence of the solution to (\ref{MODEL}) is not appropriate. The reason why we restrict ourselves is simple: in many cases, ensuring the existence of a solution to (\ref{MODEL}) is tremendously difficult. Hence, we prohibit ourselves from assuming the existence of a solution, although under the existence assumption of the solution, the proofs of results become quite easy.\footnote{See, for example, Hosoya \cite{HOB} for detailed arguments.}

\subsection{Knowledge on Ordinary Differential Equations}
In this section, we frequently use knowledge on ordinary differential equations (ODEs). Hence, we note basic knowledge on ODEs for readers.

First, consider the following differential equation:
\begin{equation}
\dot{x}(t)=h(t,x(t)),\label{ODE}
\end{equation}
where $\dot{x}$ denotes $\frac{dx}{dt}$. We assume that $h:U\to \mathbb{R}^n$, $U\subset \mathbb{R}_+\times \mathbb{R}^n$, and the relative interior of $U$ in $\mathbb{R}_+\times \mathbb{R}^n$ is nonempty (denoted by $V$). We call a set $I\subset \mathbb{R}$ an \textbf{interval} if and only if $I$ is a convex set of $\mathbb{R}$ that includes at least two points. We say that a function $x:I\to \mathbb{R}^n$ is a \textbf{solution} to (\ref{ODE}) if and only if, 1) $I$ is an interval, 2) $x(t)$ is absolutely continuous on every compact subinterval of $I$, 3) the graph of $x(t)$ is included in $U$, and 4) $\dot{x}(t)=h(t,x(t))$ for almost all $t\in I$. Suppose that $(t^*,x^*)\in U$. If a solution $x(t)$ to (\ref{ODE}) satisfies 1) $t^*\in I$ and 2) $x(t^*)=x^*$, then $x(t)$ is called a \textbf{solution with the initial value condition} $x(t^*)=x^*$, or simply, a solution to the following differential equation:
\begin{equation}
\dot{x}(t)=h(t,x(t)),\ x(t^*)=x^*.\label{CODE}
\end{equation}
Now, suppose that $h:U\to \mathbb{R}^n$ satisfies the following requirements: 1) for every $t\in \mathbb{R}$, $x\mapsto h(t,x)$ is continuous, and 2) for every $x\in \mathbb{R}^n$, $t\mapsto h(t,x)$ is measurable. Then, we say that $h$ satisfies \textbf{Carath\'eodory's condition}. If, additionally, for a set $C\subset U$, there exists $L>0$ such that
\[\|h(t,x_1)-h(t,x_2)\|\le L\|x_1-x_2\|\]
for every $(t,x_1,x_2)$ such that $(t,x_1), (t,x_2)\in C$, then $h$ is said to be \textbf{Lipschitz in $x$ on $C$}.

The following facts are well known.
\begin{enumerate}[1)]
\item Suppose that $h$ satisfies Carath\'eodory's condition. Moreover, suppose that $(t^*,x^*)\in V$ and there exist $\varepsilon>0$ and an integrable function $r:[t^*-\varepsilon,t^*+\varepsilon]\to\mathbb{R}_+$ such that $\|h(t,x)\|\le r(t)$ for all $(t,x)\in U$ with $\|(t,x)-(t^*,x^*)\|\le \varepsilon$. Then, there exists a solution $x:I\to \mathbb{R}^n$ to (\ref{CODE}), where $I$ is relatively open in $\mathbb{R}_+$.\footnote{This result is called the Carath\'eodory--Peano existence theorem. For a proof, see Ch.2 of Coddington and Levinson \cite{CL1}.}

\item Suppose that $h$ satisfies Carath\'eodory's condition, and that for every compact set $C\subset V$, $h$ is Lipschitz in $x$ on $C$. Moreover, suppose that $(t^*,x^*)\in V$ and there exists a convex neighbourhood of $t^*$ such that $t\mapsto h(t,x^*)$ is integrable on this neighbourhood. Then, there exists a solution $x:I\to \mathbb{R}^n$ to (\ref{CODE}), where $I$ is relatively open in $\mathbb{R}_+$.

\item Suppose that $h$ satisfies Carath\'eodory's condition, and that for every compact set $C\subset V$, $h$ is Lipschitz in $x$ on $C$. Moreover, suppose that for every $(t^+,x^+)\in V$, there exists a convex neighbourhood of $t^+$ such that $t\mapsto h(t,x^+)$ is integrable on this neighbourhood. Choose any $(t^*,x^*)\in V$. Suppose that $x_1(t),x_2(t)$ are two solutions to (\ref{CODE}) such that $(t,x_i(t))\in V$ for every $t\in I_i$, where $I_i$ is the domain of $x_i(t)$. Then, $x_1(t)=x_2(t)$ for every $t\in I_1\cap I_2$.\footnote{Results 2) and 3) are known as the Carath\'eodory--Picard--Lindel\"of existence theorem. For a proof, see Section 0.4 of Ioffe and Tikhomirov \cite{IT}.}

\item Suppose that $h$ is continuous. Then, any solution $x(t)$ to (\ref{ODE}) is continuously differentiable.\footnote{See Ch.2 of Hartman \cite{HA}.}
\end{enumerate}

Next, suppose that $h$ satisfies all requirements in 3) and $(t^*,x^*)\in V$. Choose a solution $x:I\to \mathbb{R}^n$ to (\ref{CODE}). A solution $y:J\to \mathbb{R}^n$ is called an \textbf{extension} of $x$ if and only if 1) $I\subset J$, and 2) $y(t)=x(t)$ for all $t\in I$. Then, $x(t)$ is said to be \textbf{nonextendable} if and only if there is no extension except $x(t)$ itself.

The following facts are well known.\footnote{Fact 5) can be proved easily. The proof of 6) is in Ch.2 of Coddington and Levinson \cite{CL1}.}
\begin{enumerate}[1)]
\setcounter{enumi}{4}
\item In addition to the requirements of 3), suppose that $V=U$. Then, there uniquely exists a nonextendable solution $x(t)$ to (\ref{CODE}). Moreover, the domain $I$ of $x(t)$ is relatively open in $\mathbb{R}_+$.

\item Suppose that all requirements of 5) hold, and let $x:I\to \mathbb{R}^n$ be the nonextendable solution to (\ref{CODE}). Choose any compact set $C\subset V$. Then, there exists $t^+\in I$ such that if $t^+\le t\in I$, then $(t,x(t))\notin C$.
\end{enumerate}
Finally, suppose that $h(t,x)=a(t)x+b(t)$, where $a(t), b(t)$ are locally integrable functions defined on $\mathbb{R}_+$ and $a(t)$ is bounded. Then, the solution to (\ref{ODE}) is determined by the following formula:
\begin{equation}
x(t)=e^{\int_0^ta(\tau)d\tau}\left[x(0)+\int_0^te^{-\int_0^sa(\tau)d\tau}b(s)ds\right].\label{LODE}
\end{equation}
This is called the formula of the solution for linear ODEs.

\subsection{Basic Lemma and Propositions}
In this subsection, we introduce a lemma and several propositions that have been proved in Hosoya \cite{HOA}. Because all results are proved in Hosoya \cite{HOA}, we omit the proofs.

\begin{lem}\label{Lemma3}
Consider the following two ODEs:
\begin{equation}
\dot{k}(t)=h_i(t,k(t)),\ k(0)=\bar{k}_i,\label{eq:eq13}
\end{equation}
where $i\in \{1,2\}$. Suppose that each $h_i$ is a real-valued function defined on some convex neighbourhood $U\subset \mathbb{R}_+\times\mathbb{R}$ of $(0,\bar{k}_i)$ and $h_i$ satisfies Carath\'eodory's condition. Then, the following results hold.
\begin{enumerate}[i)]
\item For some $i\in \{1,2\}$, if there exists a locally integrable function $r(t)$ such that
\[\sup_{k:(t,k)\in U}|h_i(t,k)|\le r(t),\]
then there exists $T>0$ such that this equation $(\ref{eq:eq13})$ has a solution $k_i:[0,T]\to\mathbb{R}$. Moreover, if $h_i(t,k)$ is continuous, then $k_i(t)$ is continuously differentiable.

\item Suppose that $\bar{k}_1\le \bar{k}_2$, $h_1(t,k)\le h_2(t,k)$ for every $(t,k)\in U$, and for some $i^*\in \{1,2\}$, $h_{i^*}$ is Lipschitz in $k$ on $U$. Suppose also that $h_{i^*}(t,\bar{k})$ is locally integrable, and there exist solutions $k_i:[0,T]\to \mathbb{R}$ to the above equations for $i\in \{1,2\}$. Then, $k_1(t)\le k_2(t)$ for all $t\in [0,T]$.\footnote{If such an $L>0$ is absent, then this lemma does not hold. For example, consider $\bar{k}_1=\bar{k}_2=0,\ h_1(t,k)=\sqrt{|k|}-\frac{t}{8},\ h_2(t,k)=\sqrt{|k|},\ k_1(t)=\frac{t^2}{16},$ and $k_2(t)\equiv 0$.

Note that, if $\bar{k}_1=\bar{k}_2$ and $h_1=h_2$, then this claim immediately implies the uniqueness of the solution. From this perspective, this lemma is an extension of the Carath\'eodory--Picard--Lindel\"of uniqueness result in the theory of ODEs.}
\end{enumerate}
\end{lem}

\begin{prop}\label{Prop1}
 Suppose that Assumptions 1-4 hold. Then, there exists a positive continuous function $c^*(p,k)$ defined on $\mathbb{R}^2_{++}$ such that, for all $(p,k)\in \mathbb{R}^2_{++}$,
\[F(k,c^*(p,k))p+u(c^*(p,k),k)=\sup_{c\ge 0}\{F(k,c)p+u(c,k)\}.\]
\end{prop}

\begin{prop}\label{Prop2}
Suppose that Assumptions 1-3 hold. Then, $\bar{V}(\bar{k})>-\infty$ for every $\bar{k}>0$ and the function $\bar{V}$ is nondecreasing and concave.
\end{prop}

By Proposition \ref{Prop2}, $\bar{V}(k)\in\mathbb{R}$ for some $k>0$ if and only if $\bar{V}(k)\in \mathbb{R}$ for every $k>0$. We say that $\bar{V}$ is \textbf{finite} if $\bar{V}(k)\in\mathbb{R}$ for every $k>0$.

\begin{prop}\label{Prop3}
Suppose that Assumptions 1-4 hold. If the value function $\bar{V}$ is finite, then it is increasing, and it is a viscosity solution to the HJB equation.
\end{prop}

\subsection{Necessity of the HJB Equation}
Because Proposition \ref{Prop3} requires the finiteness of the value function, we want an additional assumption that ensures the finiteness of the value function. However, we cannot assume that $u(c,k)$ is bounded, because a typical example of $u(c,k)$ is $\log c$. Therefore, we need an alternative condition for $u$.

First, define the {\bf constant relative risk aversion} (CRRA) function as follows:
\[u_{\theta}(x)=\begin{cases}
\frac{x^{1-\theta}-1}{1-\theta} & \mbox{if }\theta\neq 1,\\
\log x & \mbox{if }\theta=1,
\end{cases}\]
where $\theta>0$. This function is the unique solution to the following differential equation:
\[-x\frac{u''(x)}{u'(x)}=\theta,\ u(1)=0,\ u'(1)=1,\]
where the left-hand side of this equation is sometimes called the {\bf relative risk aversion} of $u$.

Using this function, we present an additional assumption.

\vspace{12pt}
\noindent
\textbf{Assumption 5}. There exist $k^*>0, c^*\ge 0, \gamma>0, \delta>0, \theta>0, a>0, b\ge 0, C\in \mathbb{R}$ such that
\begin{align}
&~(\gamma,-\delta)\in \partial F(k^*,c^*),\label{FL}\\
&~\rho-(1-\theta)\gamma>0,\label{SL}\\
&~u(c,k)\le au_{\theta}(c)+bu_{\theta}(k)+C\mbox{ for all }c>0, k>0.\label{TL}
\end{align}

\vspace{12pt}
Then, we obtain the following result.

\begin{lem}\label{Lemma4}
Suppose that Assumptions 1-5 hold. Define\footnote{Because
\[-k^*+\frac{\delta c^*+F(k^*,c^*)}{\gamma}=\frac{1}{\gamma}[\gamma(0-k^*)-\delta(0-c^*)+F(k^*,c^*)]\ge 0,\]
we have that $\hat{k},C^*>0$ for all $\bar{k}>0$.}
\[\hat{k}=\bar{k}-k^*+\frac{\delta c^*+F(k^*,c^*)}{\gamma},\ C^*=\frac{\rho-(1-\theta)\gamma}{\theta \delta}\hat{k},\]
\[V_3(\bar{k})=\begin{cases}
\frac{(C^*)^{1-\theta}\theta}{(1-\theta)(\rho-(1-\theta)\gamma)}-\frac{1}{\rho(1-\theta)}, & \mbox{if }\theta\neq 1,\\
\frac{\log C^*}{\rho}+\frac{\gamma-\rho}{\rho^2} & \mbox{if }\theta=1
\end{cases}\]
and
\[V_4(\bar{k})=\begin{cases}
\frac{\hat{k}^{1-\theta}}{(1-\theta)(\rho-(1-\theta)\gamma)}-\frac{1}{\rho(1-\theta)} & \mbox{if }\theta\neq 1,\\
\frac{\log \hat{k}}{\rho}+\frac{\gamma}{\rho^2}, & \mbox{if }\theta=1.
\end{cases}\]

Then,
\[aV_3(\bar{k})+bV_4(\bar{k})+\frac{C}{\rho}\ge \bar{V}(\bar{k})\]
for all $\bar{k}>0$.
\end{lem}

\begin{proof}
First, consider the following problem:
\begin{align*}
\max~~~~~&~\int_0^{\infty}e^{-\rho t}u(c(t),k(t))dt\\
\mbox{subject to. }&~c(t)\in W_1,\\
&~k(t)\ge 0,\ c(t)\ge 0,\\
&~\int_0^{\infty}e^{-\rho t}u(c(t),k(t))dt\mbox{ can be defined},\\
&~\dot{k}(t)=\gamma (k(t)-k^*)-\delta (c(t)-c^*)+F(k^*,c^*)\mbox{ a.e.},\\
&~k(0)=\bar{k}.
\end{align*}
Define $A^L_{\bar{k}}$ as the set of all pairs $(k(t),c(t))$ of nonnegative functions such that $k(t)$ is absolutely continuous on every compact set, $c(t)\in W_1$, $\int_0^{\infty}e^{-\rho t}u(c(t),k(t))dt$ can be defined, $k(0)=\bar{k}$, and
\[\dot{k}(t)=\gamma (k(t)-k^*)-\delta(c(t)-c^*)+F(k^*,c^*)\]
for almost all $t\ge 0$. Let
\[V_1(\bar{k})=\sup\left\{\left.\int_0^{\infty}e^{-\rho t}u(c(t),k(t))dt\right|(k(t),c(t))\in A^L_{\bar{k}}\right\}.\]

\begin{step}
$V_1(\bar{k})\ge \bar{V}(\bar{k})$ for all $\bar{k}>0$.
\end{step}

\begin{proof}[{\bf Proof of Step 1}]
By Proposition \ref{Prop2}, $\bar{V}(\bar{k})>-\infty$. For every $\varepsilon>0$ and $N>0$, there exists $(k(t),c(t))\in A_{\bar{k}}$ such that $c(t)$ is bounded and
\[\int_0^{\infty}e^{-\rho t}u(c(t),k(t))dt\ge \min\{\bar{V}(\bar{k})-\varepsilon,N\}.\]
Consider the following differential equation:
\[\dot{k}(t)=\gamma (k(t)-k^*)-\delta (c(t)-c^*)+F(k^*,c^*),\ k(0)=\bar{k}.\]
The solution to the above equation is
\[\hat{k}(t)=e^{\gamma t}\left[\bar{k}-\int_0^te^{-\gamma s}(\gamma k^*+\delta(c(s)-c^*)-F(k^*,c^*))ds\right].\]
Because $(\gamma,-\delta)\in \partial F(k^*,c^*)$, $F(k,c)\le \gamma (k-k^*)-\delta (c-c^*)+F(k^*,c^*)$ for all $(k,c)$, and thus, by Lemma \ref{Lemma3}, $\hat{k}(t)\ge k(t)$ for every $t\ge 0$. Therefore, $(\hat{k}(t),c(t))\in A^L_{\bar{k}}$, and thus,
\[V_1(\bar{k})\ge \min\{\bar{V}(\bar{k})-\varepsilon,N\}.\]
Because $\varepsilon, N$ are arbitrary, we have that $V_1(\bar{k})\ge \bar{V}(\bar{k})$. This completes the proof of Step 1.
\end{proof}

Second, consider the following problem: 
\begin{align*}
\max~~~~~&~\int_0^{\infty}e^{-\rho t}[au_{\theta}(c(t))+bu_{\theta}(k(t))]dt\\
\mbox{subject to. }&~c(t)\in W_1,\\
&~k(t)\ge 0,\ c(t)\ge 0,\\
&~\int_0^{\infty}e^{-\rho t}[au_{\theta}(c(t))+bu_{\theta}(k(t))]dt\mbox{ can be defined},\\
&~\dot{k}(t)=\gamma (k(t)-k^*)-\delta(c(t)-c^*)+F(k^*,c^*)\mbox{ a.e.},\\
&~k(0)=\bar{k}.
\end{align*}
Define $A^{L2}_{\bar{k}}$ as the set of all pairs $(k(t),c(t))$ of nonnegative functions such that $k(t)$ is absolutely continuous in every compact set, $c(t)\in W_1$, $\int_0^{\infty}e^{-\rho t}[au_{\theta}(c(t))+bu_{\theta}(k(t))]dt$ can be defined, $k(0)=\bar{k}$, and
\[\dot{k}(t)=\gamma(k(t)-k^*)-\delta(c(t)-c^*)+F(k^*,c^*)\]
for almost all $t\ge 0$. Let
\[V_2(\bar{k})=\sup\left\{\left.\int_0^{\infty}e^{-\rho t}[au_{\theta}(c(t))+bu_{\theta}(k(t))]dt\right|(k(t),c(t))\in A^{L2}_{\bar{k}}\right\}.\]

\begin{step}
$V_2(\bar{k})+\frac{C}{\rho}\ge V_1(\bar{k})$ for all $\bar{k}>0$.
\end{step}

\begin{proof}[{\bf Proof of Step 2}]
By Proposition \ref{Prop2} and Step 1, $V_1(\bar{k})>-\infty$. For every $\varepsilon>0$ and $N>0$, there exists $(k(t),c(t))\in A^L_{\bar{k}}$ such that
\[\int_0^{\infty}e^{-\rho t}u(c(t),k(t))dt\ge \min\{V_1(\bar{k})-\varepsilon,N\}.\]
Because
\[\int_0^{\infty}e^{-\rho t}\min\{au_{\theta}(c(t))+bu_{\theta}(k(t))+C,0\}dt\ge \int_0^{\infty}e^{-\rho t}\min\{u(c(t),k(t)),0\}dt>-\infty,\]
we have that
\[\int_0^{\infty}e^{-\rho t}[au_{\theta}(c(t))+bu_{\theta}(k(t))]dt\]
can be defined. Hence, $(k(t),c(t))\in A^{L2}_{\bar{k}}$, and thus,
\[V_2(\bar{k})+\frac{C}{\rho}\ge \min\{V_1(\bar{k})-\varepsilon,N\}.\]
Because $\varepsilon, N$ are arbitrary, $V_2(\bar{k})+\frac{C}{\rho}\ge V_1(\bar{k})$. This completes the proof of Step 2.
\end{proof}

\begin{step}
$aV_3(\bar{k})+bV_4(\bar{k})\ge V_2(\bar{k})$ for all $\bar{k}>0$.
\end{step}

\begin{proof}[{\bf Proof of Step 3}]
Let
\[c^*(t)=C^*e^{\frac{\gamma-\rho}{\theta}t},\]
\[k^*(t)=e^{\gamma t}\left[\bar{k}-\int_0^te^{-\gamma s}[\gamma k^*+\delta(c^*(s)-c^*)-F(k^*,c^*)]ds\right].\]
Note that,
\[\dot{k}^*(t)=\gamma (k^*(t)-k^*)-\delta (c^*(t)-c^*)+F(k^*,c^*),\]
\[\frac{d}{dt}(u_{\theta}'(c^*(t)))=(\rho-\gamma)u_{\theta}'(c^*(t)).\]
Then, for every $(k(t),c(t))\in A^{L2}_{\bar{k}}$,
\begin{align*}
&~\int_0^Te^{-\rho t}(u_{\theta}(c(t))-u_{\theta}(c^*(t)))dt\\
\le&~\int_0^Te^{-\rho t}u_{\theta}'(c^*(t))(c(t)-c^*(t))dt\\
=&~\delta^{-1}\int_0^Te^{-\rho t}u_{\theta}'(c^*(t))[\gamma (k(t)-k^*(t))-(\dot{k}(t)-\dot{k}^*(t))]dt\\
=&~\delta^{-1}\int_0^T\frac{d}{dt}[e^{-\rho t}u_{\theta}'(c^*(t))(k^*(t)-k(t))]dt\\
=&~\delta^{-1}e^{-\rho T}u_{\theta}'(c^*(T))(k^*(T)-k(T))\\
\le&~\delta^{-1}e^{-\rho T}u_{\theta}'(c^*(T))k^*(T)\\
=&~\delta^{-1}(C^*)^{-\theta}\left[\bar{k}-\left(k^*-\frac{\delta c^*+F(k^*,c^*)}{\gamma}\right)(1-e^{-\gamma T})-\frac{\theta \delta C^*}{\rho-(1-\theta)\gamma}(1-e^{\frac{(1-\theta)\gamma-\rho}{\theta}T})\right]\\
\to&~0\mbox{ (as $T\to \infty$)}.
\end{align*}
Moreover, because $\gamma k^*-\delta c^*-F(k^*,c^*)\le 0$,
\begin{align*}
k(t)=&~e^{\gamma t}\left(\bar{k}-\int_0^te^{-\gamma s}[\gamma k^*+\delta(c(s)-c^*)-F(k^*,c^*)]ds\right)\\
\le&~e^{\gamma t}\left(\bar{k}-\int_0^te^{-\gamma s}[\gamma k^*-\delta c^*-F(k^*,c^*)]ds\right)\\
\le&~e^{\gamma t}\left(\bar{k}-\int_0^{\infty}e^{-\gamma s}[\gamma k^*-\delta c^*-F(k^*,c^*)]ds\right)\\
=&~\left(\bar{k}-k^*+\frac{\delta c^*+F(k^*,c^*)}{\gamma}\right)e^{\gamma t}=\hat{k}e^{\gamma t},
\end{align*}
and thus,
\begin{align*}
&~\int_0^{\infty}e^{-\rho t}[au_{\theta}(c(t))+bu_{\theta}(k(t))]dt\\
\le&~a\int_0^{\infty}e^{-\rho t}u_{\theta}(c^*(t))dt+b\int_0^{\infty}e^{-\rho t}u_{\theta}(\hat{k}e^{\gamma t})dt\\
=&~aV_3(\bar{k})+bV_4(\bar{k}),
\end{align*}
which completes the proof.
\end{proof}

Steps 1-3 show that our claim is correct. This completes the proof.

\end{proof}
\setcounter{step}{0}

Combining Proposition \ref{Prop3} and Lemma \ref{Lemma4}, we obtain the following result.

\begin{thm}\label{Theorem1}
Suppose that Assumptions 1-5 hold. Then $\bar{V}\in \mathscr{V}$, and $\bar{V}$ is a viscosity solution to the HJB equation.
\end{thm}

\begin{proof}
By Lemma \ref{Lemma4}, we have that $\bar{V}$ is finite, and by Propositions \ref{Prop2} and \ref{Prop3}, we have that $\bar{V}$ is an increasing and concave viscosity solution to the HJB equation. It suffices to show that $\bar{V}$ satisfies the growth condition (\ref{GC}).

By Lemma \ref{Lemma1}, $\inf_{t\ge 0}k^+(t,\bar{k})>0$ for all $\bar{k}>0$, and thus it suffices to show that for $i\in \{3,4\}$,
\[\limsup_{t\to\infty}e^{-\rho t}V_i(k^+(t,\bar{k}))\le 0,\]
Define
\[\hat{k}(t)=\hat{k}e^{\gamma t}.\]
By Lemma \ref{Lemma3} and the calculation in Step 3 of the proof of Lemma \ref{Lemma4}, $\hat{k}(t)\ge k^+(t,\bar{k})$ for all $t\ge 0$, and it suffices to show that, for $i\in \{3,4\}$,
\[\lim_{t\to \infty}e^{-\rho t}V_i(\hat{k}(t))=0.\]
If $\theta=1$, then 
\[V_i(\hat{k}(t))=A\log (e^{\gamma t}+B)+C\]
for some $A>0,\ B\ge 0$, and $C\in\mathbb{R}$, and
\[e^{-\rho t}V_i(\hat{k}(t))=Ae^{-\rho t}\log (e^{\gamma t}+B)+e^{-\rho t}C\to 0\]
as $t\to \infty$. If $\theta\neq 1$, then
\[V_i(\hat{k}(t))=A(e^{\gamma t}+B)^{1-\theta}+C,\]
for some $A, B, C\in\mathbb{R}$ such that $A(1-\theta)>0$ and $B\ge 0$. If $\theta<1$, then
\[e^{-\rho t}V_i(\hat{k}(t))=A(e^{\frac{(1-\theta)\gamma-\rho}{1-\theta}t}+e^{-\frac{\rho}{1-\theta}t}B)^{1-\theta}+e^{-\rho t}C\to 0\]
as $t\to\infty$. If $\theta>1$, then
\[e^{-\rho t}V_i(\hat{k}(t))=e^{-\rho t}[A(e^{\gamma t}+B)^{1-\theta}+C]\to 0\]
as $t\to \infty$. Thus, in any case, our claim is correct. This completes the proof.
\end{proof}

\subsection{Differential Inclusions}
In the previous subsection, we showed that under Assumptions 1-5, the value function is a viscosity solution to the HJB equation. We require the converse: that is, we need to show that, under some additional assumption, $\bar{V}$ is the unique viscosity solution to the HJB equation in $\mathscr{V}$. However, to derive this result, we need the aid of differential inclusions.

Hence, in this subsection, we introduce several properties of differential inclusions. First, consider the following autonomous differential inclusion:
\begin{equation}\label{DI}
\dot{k}(t)\in \Gamma(k(t)),\ k(0)=\bar{k}>0,
\end{equation}
where $\Gamma:\mathbb{R}_{++}\twoheadrightarrow \mathbb{R}$ is a nonempty-valued set function. A function $k:I\to \mathbb{R}_{++}$ is called a {\bf solution} to (\ref{DI}) if and only if, 1) $I$ is an interval that includes $0$, 2) $k(t)$ is absolutely continuous on any compact subinterval of $I$, 3) $k(0)=\bar{k}$, and 4) $\dot{k}(t)\in \Gamma(k(t))$ for almost every $t\in I$. It is known that if $\Gamma$ is a compact- and convex-valued upper hemi-continuous mapping, then there exists at least one solution $k(t)$ to (\ref{DI}) defined on $[0,T]$ for some $T>0$.\footnote{If the domain of $\Gamma$ is $\mathbb{R}$, then we can apply Theorem 1 of Maruyama \cite{MA} directly. It is easy to extend this result to our case. Define $\Phi(k)=\Gamma(k)$ if $k>\bar{k}/2$ and $\Phi(k)=\Gamma(\bar{k}/2)$ otherwise. Then, $\Phi:\mathbb{R}\twoheadrightarrow \mathbb{R}$ is a nonempty-, compact-, and convex-valued upper hemi-continuous mapping, and thus, the inclusion
\[\dot{k}(t)\in \Phi(k(t)),\ k(0)=\bar{k}>0,\]
has a solution $k:[0,T]\to \mathbb{R}$ such that $k(t)\ge \bar{k}/2$ for any $t\in [0,T]$. This $k(t)$ is also a solution to (\ref{DI}).}

We need the following lemmas. The proofs of these lemmas are placed in the appendix.

\begin{lem}\label{Lemma5}
Consider the differential inclusion $(\ref{DI})$, and the differential equation
\begin{equation}\label{DEI}
\dot{k}(t)=h(k(t)),\ k(0)=\bar{k}>0,
\end{equation}
where $h:\mathbb{R}_{++}\to \mathbb{R}$ is a locally Lipschitz function such that, for any $k\in \mathbb{R}_{++}$ and $y\in \Gamma(k)$, $y\le h(k)$. Let $k_1(t)$ be a solution to $(\ref{DEI})$ and $k_2(t)$ be a solution to $(\ref{DI})$ defined on $[0,T]$. Then, $k_1(t)\ge k_2(t)$ for all $t\in [0,T]$.
\end{lem}

\begin{lem}\label{Lemma6}
Consider the differential inclusion $(\ref{DI})$, where $\Gamma:\mathbb{R}_{++}\twoheadrightarrow \mathbb{R}$ is a nonempty-, compact-, and convex-valued upper hemi-continuous mapping. Suppose that there exist two positive continuous functions $\hat{k}(t)$ and $\bar{k}(t)$ defined on $\mathbb{R}_+$ such that $\bar{k}(t)\le k(t)\le \hat{k}(t)$ for any solution $k(t)$ to $(\ref{DI})$ defined on $[0,T]$ for some $T>0$ and $t\in [0,T]$. Then, this inclusion $(\ref{DI})$ has a solution defined on $\mathbb{R}_+$ itself.
\end{lem}

\subsection{Sufficiency of the HJB Equation and Construction of the Solution}

In this subsection, we examine the sufficiency of the HJB equation to determine the value function. First, we should mention an important example. Consider the case in which $u(c,k)=-1/c$ and $F(k,c)=\sqrt{k}-c$. The corresponding HJB equation is
\[\sup_{c\ge 0}\{(\sqrt{k}-c)V'(k)-1/c\}-\rho V(k)=0,\]
and thus, $V\equiv 0$ is a classical solution to this equation. However, this function is not increasing. Note that, these $u$ and $F$ satisfies Assumptions 1-5.\footnote{For Assumption 5, choose $\theta=2$.} By Theorem 1, the value function $\bar{V}$ is an increasing solution to the HJB equation, and thus $\bar{V}\neq V$. This indicates that the usual uniqueness argument cannot be applied to our model, and, at least, the increasing requirement is crucial for the uniqueness result.

We need two additional lemmas. The proofs of these lemmas are placed in the appendix.\footnote{Although Lemma \ref{Lemma7} may be a known result, the author could not find an appropriate reference. Fortunately, the proof of this result is relatively easy and short, and thus we put the proof into the appendix.}

\begin{lem}\label{Lemma7}
Suppose that Assumptions 1-4 hold, and that $V:\mathbb{R}_{++}\to \mathbb{R}$ is a concave function. Then, $V$ is a viscosity solution to the HJB equation if and only if, for any $k>0$ and $p\in \partial V(k)$,
\begin{equation}\label{EQQ}
\sup_{c\ge 0}\{F(k,c)p+u(c,k)\}=\rho V(k).
\end{equation}
\end{lem}

\begin{lem}\label{Lemma8}
Suppose that $U\subset \mathbb{R}$ is an open interval, $H:U\to \mathbb{R}$ is concave, $A<B$, and $x:[A,B]\to U$ is absolutely continuous. Let $\psi(t)=H(x(t))$. Then, $\psi(t)$ is also absolutely continuous. Moreover, if both $\dot{\psi}(t)$ and $\dot{x}(t)$ are defined, then for any $p\in \partial H(x(t))$, $\dot{\psi}(t)=p\dot{x}(t)$.
\end{lem}

We now introduce an additional assumption.

\vspace{12pt}
\noindent
\textbf{Assumption 6}. There exists $\varepsilon_0>0$ such that $\frac{\partial u}{\partial c}$ is continuously differentiable in $k$ on $\mathbb{R}_{++}\times ]0,\varepsilon_0[$, and there exists a continuously differentiable function $H:]0,\varepsilon_0[\times \mathbb{R}_{++}\to \mathbb{R}$ such that $\frac{\partial H}{\partial k}(k,c)>0,\ \frac{\partial H}{\partial c}(k,c)<0$ for all $(k,c)\in ]0,\varepsilon_0[\times \mathbb{R}_{++}$, and if $k<\varepsilon_0$ and $H(k,c)\neq 0$, then $F$ is continuously differentiable in $c$ and $\frac{\partial F}{\partial c}$ is continuously differentiable in $k$ around $(k,c)$. Moreover, there exists $k>0$ such that $\inf_{c\ge 0}D_{k,+}F(k,c)>\rho$.\footnote{The last inequality is similar to assumption (A5) of Frankowska et al. \cite{FZZ}. Because these two assumptions are used in a similar way, there may be some hidden relationship between them. However, we could not identify it.}

\vspace{12pt}
We define some additional notation. Let $B_{\bar{k}}$ denote the set of all pairs of nonnegative functions $(k(t),c(t))$ defined on $\mathbb{R}_+$ such that $k(t)$ is absolutely continuous on any compact interval, $c(t)\in W$, $\lim_{T\to \infty}\int_0^Te^{-\rho t}u(c(t),k(t))dt$ exists, $k(0)=\bar{k}$, and
\[\dot{k}(t)=F(k(t),c(t))\]
for almost all $t\ge 0$. Clearly, $A_{\bar{k}}\subset B_{\bar{k}}$, but it is unknown whether $A_{\bar{k}}=B_{\bar{k}}$.

Now, suppose that $V:\mathbb{R}_{++}\to \mathbb{R}$ is an increasing and concave function. As we discussed in Subsection 2.4, $\partial V$ is a nonempty-, compact-, and convex-valued upper hemi-continuous mapping. Therefore, the mapping $F(k,c^*(\partial V(k),k))$ is also a nonempty-, compact-, and convex-valued upper hemi-continuous mapping. Note also that, the mapping $c^*(\partial V(k),k)$ is also compact-valued and upper hemi-continuous, and thus, this mapping is measurable, and for any continuous function $k:I\to \mathbb{R}_{++}$ defined on an interval $I$, there is a measurable selection $c(t)$ of $c^*(\partial V(k(t)),k(t))$.\footnote{See Section 8.1 of Ioffe and Tikhomirov \cite{IT}.}

The next proposition is crucial for our next main result.

\begin{prop}\label{Prop4}
Suppose that Assumptions 1-4 and 6 hold. Suppose also that $V\in \mathscr{V}$ is a viscosity solution to the HJB equation. Choose any $\bar{k}>0$, and consider the following differential inclusion:
\begin{equation}
\dot{k}(t)\in F(k(t),c^*(\partial V(k(t)),k(t))),\ k(0)=\bar{k}.\label{SOL}
\end{equation}
Then, there exists a solution $k^*(t)$ to the above equation defined on $\mathbb{R}_+$, and for any such solution, $\inf_{t\ge 0}k^*(t)>0$. Moreover, if we define
\begin{equation}\label{CONSUM}
c^*(t)=\arg\min\{|\limsup_{n\to \infty}n(k^*(t+n^{-1})-k^*(t))-F(k^*(t),c)||c\in c^*(\partial V(k^*(t)),k^*(t))\},
\end{equation}
then, $c^*(t)$ is locally bounded, and $(k^*(t),c^*(t))\in B_{\bar{k}}$. Furthermore, for any such pair $(k^*(t),c^*(t))$,
\[V(\bar{k})=\lim_{T\to \infty}\int_0^Te^{-\rho t}u(c^*(t),k^*(t))dt,\]
and for every $(k(t),c(t))\in B_{\bar{k}}$, if $\inf_{t\ge 0}k(t)>0$, then 
\[\lim_{T\to \infty}\int_0^Te^{-\rho t}u(c(t),k(t))dt\le V(\bar{k}).\]
\end{prop}

\begin{proof}

We separate the proof into three steps.

\begin{step}
Suppose that $V\in\mathscr{V}$ is a viscosity solution to the HJB equation. Then, there exists a solution $k^*(t)$ to the differential inclusion $(\ref{SOL})$ defined on $\mathbb{R}_+$. Moreover, for any such solution, $\inf_{t\ge 0}k^*(t)>0$.
\end{step}

\begin{proof}[{\bf Proof of Step 1}]
First, as we have mentioned, the differential inclusion (\ref{SOL}) has at least one solution $k(t)$ defined on $[0,T]$.

Let $\varepsilon>0$ satisfy that
\[\inf_{c\ge 0}D_{k,+}F(\varepsilon,c)>\rho,\ \varepsilon<\min\{\bar{k},\varepsilon_0\}.\]
We show that, for any solution $k(t)$ to (\ref{SOL}) defined on $[0,T]$, $k(t)\ge \varepsilon$.

Suppose that this claim is incorrect. Then, there exists $t^+>0$ such that $0<k(t^+)<\varepsilon$. Let $I\subset [0,T[$ be the set of all $t$ such that 1) $k(t)<\varepsilon$, 2) $\dot{k}(t)<0$ and (\ref{SOL}) holds at $t$, 3) $V$ is differentiable at $k(t)$, and 4) there exists an Alexandrov Hessian $L\in \mathbb{R}$ of $V$ at $k(t)$.\footnote{For a concave function $f$, the number $L$ is called an {\bf Alexandrov Hessian} of $f$ at $x$ if and only if, for all $\varepsilon'>0$, there exists $\delta'>0$ such that if $|x'-x|<\delta'$, then $|y'-y-L(x'-x)|\le \varepsilon|x'-x|$ for all $y\in \partial f(x),\ y'\in \partial f(x')$. For a detailed argument, see Alexandrov \cite{AL} or Howard \cite{HO}.} We show that $I$ is nonempty. First, we show that if $F(k,c^*(\partial V(k),k))\neq \{0\}$, then $V$ is differentiable at $k$. Suppose not. Then, $\partial V(k)=[D_+V(k),D_-V(k)]$, where $D_+V(k)<D_-V(k)$. Choose $p\in ]D_+V(k),D_-V(k)[$. By Lemma \ref{Lemma7},
\begin{align*}
F(k,c^*(p,k))p+u(c^*(p,k),k)=&~\rho V(k),\\
F(k,c^*(p,k))q+u(c^*(p,k),k)\le&~\rho V(k)
\end{align*}
for any $q\in \partial V(k)$. Therefore, $F(k,c^*(p,k))=0$. By the continuity of $F$ and $c^*$, we have that $F(k,c^*(\partial V(k),k))=\{0\}$, which is a contradiction. Therefore, $V$ is differentiable at $k$. Now, it is easy to show that there exists $t^*\in [0,T[$ such that $k(t^*)<\varepsilon$, $\dot{k}(t^*)<0$, and $\dot{k}(t^*)\in F(k(t^*),c^*(\partial V(k(t^*)),k(t^*)))$. This implies that $F(k(t^*),c^*(\partial V(k(t^*)),k(t^*)))\neq \{0\}$, and thus $V$ is differentiable at $k(t^*)$. Because $\partial V$ is upper hemi-continuous, there exists a neighbourhood $U$ of $k(t^*)$ such that if $k\in U$, then $F(k,c^*(\partial V(k),k))\neq \{0\}$, which implies that $V$ is differentiable around $k(t^*)$. Again, because $\partial V$ is upper hemi-continuous, we have that $V$ is continuously differentiable around $k(t^*)$. Therefore, there exists an open neighbourhood of $t^*$ such that $k(t)$ is a solution to the following differential equation:
\[\dot{k}(t)=F(k(t),c^*(V'(k(t)),k(t))),\]
which implies that $k(t)$ is continuously differentiable around $t^*$. Because of Rademacher's theorem and Alexandrov's theorem, we have that for almost all $t$ near to $t^*$, $\dot{k}(t)<0$, $V$ is differentiable at $k(t)$, and there exists an Alexandrov Hessian $L\in \mathbb{R}$ of $V$ at $k(t)$, as desired. Moreover, we have shown that if $t\in I$, then there exists an open neighbourhood $U$ of $t$ and an open neighbourhood $W$ of $k(t)$ such that $V$ is continuously differentiable on $W$, $k(\cdot)$ is continuously differentiable on $U$, and $t'\in I$ for almost all $t'\in U$.

Choose $c(T)\in c^*(\partial V(k(T)),k(T))$, and if $0\le t<T$, define
\[c(t)=\arg\min\{|\limsup_{n\to \infty}n(k(t+n^{-1})-k(t))-F(k(t),c)||c\in c^*(\partial V(k(t)),k(t))\}.\]
Then, $c(t)$ is a measurable, positive, and locally bounded function defined on $[0,T]$, and $c(t)\in c^*(\partial V(k(t)),k(t))$ for all $t\in [0,T]$. Note that, if $t\in I$, then $\partial V(k(t))=\{V'(k(t))\}$, and thus $c(t)$ is continuous at $t$.

Choose $t^*\in I$. Because $\dot{k}(t^*)<0$, for any sufficiently small $h>0$, $k(t^*+h)<k(t^*)$ and $V$ is differentiable on $[k(t^*+h),k(t^*)]$. Thus, for any $t\in [t^*,t^*+h]$,
\[F(k(t),c(t))V'(k(t))+u(c(t),k(t))=\rho V(k(t)),\]
and there exist $k_1,k_2\in [k(t^*+h),k(t^*)],\ \theta_1,\theta_2\in ]0,1[$, $p\in \partial_kF(k_2,c(t^*+h))$, and $q\in \partial_cF(k(t^*),c(t^*+\theta_2 h))$ such that,\footnote{See the mean value theorem in Subsection 2.4.}
\begin{align*}
&~\rho V'(k_1)(k(t^*+h)-k(t^*))=\rho (V(k(t^*+h))-V(k(t^*)))\\
=&~F(k(t^*+h),c(t^*+h))V'(k(t^*+h))+u(c(t^*+h),k(t^*+h))\\
&~-F(k(t^*),c(t^*))V'(k(t^*))-u(c(t^*),k(t^*))\\
=&~(F(k(t^*+h),c(t^*+h))-F(k(t^*),c(t^*+h)))V'(k(t^*+h))\\
&~+(F(k(t^*),c(t^*+h))-F(k(t^*),c(t^*)))V'(k(t^*+h))\\
&~+F(k(t^*),c(t^*))(V'(k(t^*+h))-V'(k(t^*)))\\
&~+u(c(t^*+h),k(t^*+h))-u(c(t^*+h),k(t^*))\\
&~+u(c(t^*+h),k(t^*))-u(c(t^*),k(t^*))\\
=&~pV'(k(t^*+h))(k(t^*+h)-k(t^*))+\dot{k}(t^*)(V'(k(t^*+h))-V'(k(t^*)))\\
&~+\frac{\partial u}{\partial k}(c(t^*+h),k(t^*+\theta_1h))(k(t^*+h)-k(t^*))\\
&~+\left(\frac{\partial u}{\partial c}(c(t^*+\theta_2h),k(t^*+\theta_2h))+qV'(k(t^*+h))\right)(c(t^*+h)-c(t^*))\\
&~+\left(\frac{\partial u}{\partial c}(c(t^*+\theta_2h),k(t^*))-\frac{\partial u}{\partial c}(c(t^*+\theta_2h),k(t^*+\theta_2h))\right)(c(t^*+h)-c(t^*)).
\end{align*}
To modify this equation,
\begin{align}
&~\frac{(\rho V'(k_1)-pV'(k(t^*+h)))(k(t^*+h)-k(t^*))}{h}\nonumber \\
=&~\frac{\dot{k}(t^*)(V'(k(t^*+h))-V'(k(t^*)))}{h}\nonumber \\
&~+\frac{\frac{\partial u}{\partial k}(c(t^*+h),k(t^*+\theta_1h))(k(t^*+h)-k(t^*))}{h}\label{EVA2}\\
&~+\frac{\left(\frac{\partial u}{\partial c}(c(t^*+\theta_2h),k^*(t^*+\theta_2h))+qV'(k(t^*+h))\right)(c(t^*+h)-c(t^*))}{h}\nonumber \\
&~+\frac{\left(\frac{\partial u}{\partial c}(c(t^*+\theta_2h),k^*(t^*))-\frac{\partial u}{\partial c}(c(t^*+\theta_2h),k^*(t^*+\theta_2h))\right)(c(t^*+h)-c(t^*))}{h}.\nonumber
\end{align}
Because $k(t^*+h)\le k(t^*)<\varepsilon$, $p>\inf_{c\ge 0}D_{k,+}F(k(t^*),c)>\rho$, and thus
\begin{align*}
&~\liminf_{h\downarrow 0}\frac{(\rho V'(k_1)-pV'(k(t^*+h)))(k(t^*+h)-k(t^*))}{h}\\
\ge&~(\rho-\inf_{c\ge 0}D_{k,+}F(k(t^*),c))V'(k(t^*))\dot{k}(t^*)>0.
\end{align*}
On the other hand, the first and second terms of the right-hand side of (\ref{EVA2}) are always nonpositive. Because $\frac{\partial u}{\partial c}$ is differentiable in $k$ on $\mathbb{R}_{++}\times ]0,\varepsilon_0[$, the fourth term of the right-hand side of (\ref{EVA2}) is
\[\frac{\partial^2u}{\partial k\partial c}(c(t^*+\theta_2h),k(t^*+\theta''h))\frac{k(t^*)-k(t^*+\theta_2h)}{h}(c(t^*+h)-c(t^*))\]
for some $\theta''\in [0,\theta_2]$. Therefore, the fourth term converges to $0$ as $h\to 0$.

We claim that there exists $t^*\in I$ such that the limsup of the third term is not greater than $0$: that is, we show that there exists $t^*\in I$ such that
\begin{equation}\label{THIRD}
\limsup_{h\downarrow 0}\frac{\left(\frac{\partial u}{\partial c}(c(t^*+\theta_2h),k(t^*+\theta_2h))+qV'(k(t^*+h))\right)(c(t^*+h)-c(t^*))}{h}\le 0.
\end{equation}
Because of the definition of $c(t)$ and the first-order condition,
\[\frac{\partial u}{\partial c}(c(t^*+\theta_2h),k(t^*+\theta_2h))=-rV'(k(t^*+\theta_2h)),\]
where $r\in \partial_cF(k(t^*+\theta_2h),c(t^*+\theta_2h))$.

We separate the proof into two cases. First, suppose that there exists $t^*\in I$ such that $H(k(t^*),c(t^*))\neq 0$. Then, $F$ is differentiable in $c$ and $\frac{\partial F}{\partial c}$ is differentiable in $k$ around $(k(t^*),c(t^*))$. Hence, if $h>0$ is sufficiently small, then $q=\frac{\partial F}{\partial c}(k(t^*),c(t^*+\theta_2h))$ and $r=\frac{\partial F}{\partial c}(k(t^*+\theta_2h),c(t^*+\theta_2h))$. In this case, the absolute value of the third term of the right-hand side of (\ref{EVA2}) is bounded from
\begin{align*}
&~\frac{|q|(V'(k(t^*+h))-V'(k(t^*)))\times |c(t^*+h)-c(t^*)|}{h}\\
&~+\frac{|q-r|V'(k(t^*+\theta_2h))\times |c(t^*+h)-c(t^*)|}{h},
\end{align*}
where
\[\frac{V'(k(t^*+h))-V'(k(t^*))}{h}\to L\dot{k}(t^*)\mbox{ as }h\downarrow 0\]
and
\[\frac{q-r}{h}=-\frac{\partial^2 F}{\partial k\partial c}(k(t^*+\theta'h),c(t^*+\theta_2h))\frac{k(t^*+\theta_2h)-k(t^*)}{h}\]
for some $\theta'\in [0,\theta_2]$. Therefore, (\ref{THIRD}) holds. Next, suppose that there is no $t^*\in I$ such that $H(k(t^*),c(t^*))\neq 0$. Then, there exists $h>0$ such that $H(k(t),c(t))=0$ for all $t\in [t^*,t^*+h]$. By the implicit function theorem, there exists a continuously differentiable function $c^+(k)$ such that $c(t)=c^+(k(t))$ for all $t\in [t^*,t^*+h]$. This implies that $c(t)$ is continuously differentiable around $t^*$. Because $(c^+)'(k)>0$, we have that $c(t)$ is decreasing around $t^*$. In this case, the value of the third term of the right-hand side of (\ref{EVA2}) is bounded from
\begin{align*}
&~\frac{q(V'(k(t^*+h))-V'(k(t^*)))(c(t^*+h)-c(t^*))}{h}\\
&~+\frac{(q-r)V'(k(t^*+\theta_2h))(c(t^*+h)-c(t^*))}{h}.
\end{align*}
The first term of the above formula converges to $0$ because of the same arguments as above. Moreover, $F$ is continuously differentiable at $(k(t^*),c(t^*+\theta_2h))$, and
\[q=\frac{\partial F}{\partial c}(k(t^*),c(t^*+\theta_2h))\to D_{c,-}F(k(t^*),c(t^*))\mbox{ as }h\downarrow 0.\]
Because $\theta_2<1$,
\[r\le D_{c,+}F(k(t^*+\theta_2h),c(t^*+h)).\]
We show that
\[\limsup_{h\downarrow 0}D_{c,+}F(k(t^*+\theta_2h),c(t^*+h))\le D_{c,-}F(k(t^*),c(t^*)).\]
Suppose not. Then, there exist $\delta>0$ and a sequence $(k_m,c_m)$ such that $(k_m,c_m)\to (k(t^*),c(t^*))$ as $m\to \infty$, and $D_{c,+}F(k_m,c_m)>D_{c,-}F(k(t^*),c(t^*))+\delta$ for all $m$. Then, there exists $a>0$ such that
\[\frac{F(k(t^*),c(t^*)-a)-F(k(t^*),c(t^*))}{-a}<D_{c,-}F(k(t^*),c(t^*))+\delta.\]
Therefore, for sufficiently large $m$,
\[D_{c,+}F(k_m,c_m)\le \frac{F(k_m,c_m-a)-F(k_m,c_m)}{-a}<D_{c,-}F(k(t^*),c(t^*))+\delta,\]
which is a contradiction. Combining the above inequalities,
\[\limsup_{h\downarrow 0}\frac{(q-r)V'(k(t^*+\theta_2h))(c(t^*+h)-c(t^*))}{h}\le 0,\]
and thus, (\ref{THIRD}) holds. This implies that the limsup of the right-hand side of (\ref{EVA2}) is nonpositive, which is a contradiction.

Hence, $k(t)\ge \varepsilon$ for every $t\in [0,T]$. By Lemma \ref{Lemma5}, $k(t)\le k^+(t,\bar{k})$ for every $t\in [0,T]$, and by Lemma \ref{Lemma6}, there exists a solution $k^*(t)$ to (\ref{SOL}) defined on $\mathbb{R}_+$. Clearly, $k^*(t)\ge \varepsilon$ for every $t\ge 0$, and thus $\inf_{t\ge 0}k^*(t)>0$. This completes the proof of Step 1.
\end{proof}

\begin{step}
Under the assumptions in Step 1, choose a solution $k^*(t)$ to equation $(\ref{SOL})$ defined on $\mathbb{R}_+$, and define $c^*(t)$ by $(\ref{CONSUM})$. Then, $(k^*(t),c^*(t))\in B_{\bar{k}}$ and
\[V(\bar{k})=\lim_{T\to \infty}\int_0^Te^{-\rho t}u(c^*(t),k^*(t))dt.\]
\end{step}

\begin{proof}[{\bf Proof of Step 2}]
By Step 1, there exists $\varepsilon>0$ such that
\[\varepsilon\le k^*(t)\le k^+(t,\bar{k})\]
for all $t\ge 0$. Moreover, by the definition of $c^*(t)$,
\[F(k^*(t),c^*(t))=\dot{k}^*(t)\]
for almost every $t\in \mathbb{R}_+$. Let
\[p^*(t)=\min\arg\min\{|\rho V(k^*(t))-u(c^*(t),k^*(t))-F(k^*(t),c^*(t))p||p\in \partial V(k^*(t))\}.\]
Then, $p^*(t)$ is measurable and positive, and
\[F(k^*(t),c^*(t))p^*(t)+u(c^*(t),k^*(t))=\rho V(k^*(t)).\]
for all $t\ge 0$. Choose any $T>0$. Let $W^*(t)=V(k^*(t))$. By Lemma \ref{Lemma8}, $W^*(t)$ is absolutely continuous on $[0,T]$ and $\dot{W}^*(t)=p^*(t)\dot{k}^*(t)$ almost everywhere, and thus,
\begin{align*}
\int_0^Te^{-\rho t}u(c^*(t),k^*(t))dt=&~\int_0^Te^{-\rho t}[F(k^*(t),c^*(t))p^*(t)+u(c^*(t),k^*(t))]dt\\
&~-\int_0^Te^{-\rho t}p^*(t)\dot{k}^*(t)dt\\
=&~-\int_0^T[-\rho e^{-\rho t}W^*(t)+e^{-\rho t}\dot{W}^*(t)]dt\\
=&~\int_0^T\frac{d}{dt}[-e^{-\rho t}W^*(t)]dt\\
=&~V(\bar{k})-e^{-\rho T}V(k^*(T)),\\
\end{align*}
and hence,
\[V(\bar{k})-e^{-\rho T}V(\varepsilon)\ge \int_0^Te^{-\rho t}u(c^*(t),k^*(t))dt\ge V(\bar{k})-e^{-\rho T}V(k^+(T,\bar{k})).\]
Therefore, by (\ref{GC}),
\[\lim_{T\to \infty}\int_0^Te^{-\rho t}u(c^*(t),k^*(t))dt=V(\bar{k}),\]
which implies that $(k^*(t),c^*(t))\in B_{\bar{k}}$. This completes the proof of Step 2.
\end{proof}

\begin{step}
If $(k(t),c(t))\in B_{\bar{k}}$ and $\inf_{t\ge 0}k(t)>0$, then
\[\lim_{T\to \infty}\int_0^Te^{-\rho t}u(c(t),k(t))dt\le V(\bar{k}).\]
\end{step}

\begin{proof}[{\bf Proof of Step 3}]
Suppose that $(k(t),c(t))\in B_{\bar{k}}$ and $\inf_{t\ge 0}k(t)>0$. Choose $T>0$. Define $W(t)=V(k(t))$. By Lemma \ref{Lemma8}, $W(t)$ is absolutely continuous on $[0,T]$ and $\dot{W}(t)=D_-V(k(t))\dot{k}(t)$ almost everywhere. Moreover, by Lemma \ref{Lemma7},
\begin{align*}
&~F(k(t),c(t))D_-V(k(t))+u(c(t),k(t))\\
\le&~\sup_{c\ge 0}\{F(k(t),c)D_-V(k(t))+u(c,k(t))\}=\rho V(k(t)).
\end{align*}
Therefore,
\begin{align*}
\int_0^Te^{-\rho t}u(c(t),k(t))dt=&~\int_0^Te^{-\rho t}[F(k(t),c(t))D_-V(k(t))+u(c(t),k(t))]dt\\
&~-\int_0^Te^{-\rho t}D_-V(k(t))\dot{k}(t)dt\\
\le&~-\int_0^T[-\rho e^{-\rho t}W(t)+e^{-\rho t}\dot{W}(t)]dt\\
=&~\int_0^T\frac{d}{dt}[-e^{-\rho t}W(t)]dt\\
=&~V(\bar{k})-e^{-\rho T}V(k(T)).
\end{align*}
Because $\inf_{t\ge 0}k(t)>0$ and $V$ satisfies (\ref{GC}), the right-hand side converges to $V(\bar{k})$ as $T\to \infty$. Therefore,
\[\lim_{T\to \infty}\int_0^Te^{-\rho t}u(c(t),k(t))dt\le V(\bar{k}),\]
as desired. This completes the proof of Step 3.
\end{proof}

Steps 1-3 state that all of our claims in Proposition \ref{Prop4} are correct. This completes the proof.
\end{proof}
\setcounter{step}{0}

We now consider the following differential inclusion:
\begin{equation}
\dot{k}(t)\in F(k(t),c^*(\partial \bar{V}(k(t)),k(t))),\ k(0)=\bar{k},\label{SOL2}
\end{equation}
and the corresponding definition of the function:
\begin{equation}\label{CONSUM2}
c^*(t)=\arg\min\{|\limsup_{n\to \infty}n(k^*(t+n^{-1})-k^*(t))-F(k^*(t),c)||c\in c^*(\partial \bar{V}(k^*(t)),k^*(t))\}.
\end{equation}

\begin{thm}\label{Theorem2}
Suppose that Assumptions 1-6 hold. Then, $\bar{V}$ is the unique viscosity solution to the HJB equation in $\mathscr{V}$. 
\end{thm}

\begin{proof}
Because Assumptions 1-6 hold, $\bar{V}\in \mathscr{V}$, and it is a viscosity solution to the HJB equation. By Proposition \ref{Prop4}, there exists a solution $k^*(t)$ to (\ref{SOL2}) defined on $\mathbb{R}_+$ such that $\inf_{t\ge 0}k^*(t)>0$, and if we define $c^*(t)$ by (\ref{CONSUM2}), then $(k^*(t),c^*(t))\in B_{\bar{k}}$ and
\[\bar{V}(\bar{k})=\lim_{T\to \infty}\int_0^Te^{-\rho t}u(c^*(t),k^*(t))dt.\]
Next, suppose that $V\in\mathscr{V}$ is also a viscosity solution to the HJB equation. Because $\inf_{t\ge 0}k^*(t)>0$, by Proposition \ref{Prop4},
\[V(\bar{k})\ge \lim_{T\to \infty}\int_0^Te^{-\rho t}u(c^*(t),k^*(t))dt=\bar{V}(\bar{k}).\]
By Proposition \ref{Prop4}, there exists a solution $k^+(t)$ to (\ref{SOL}) defined on $\mathbb{R}_+$ such that $\inf_{t\ge 0}k^+(t)>0$. Define $c^+(t)$ as (\ref{CONSUM}), where $k^*(t)$ is replaced with $k^+(t)$. Then, $(k^+(t),c^+(t))\in B_{\bar{k}}$, and
\[V(\bar{k})=\lim_{T\to \infty}\int_0^Te^{-\rho t}u(c^+(t),k^+(t))dt\le \bar{V}(\bar{k}).\]
Hence, we conclude that $V=\bar{V}$. This completes the proof.
\end{proof}

Therefore, under Assumptions 1-6, the HJB equation is the perfect characterization for the value function in the functional space $\mathscr{V}$.

In the proof of Theorem \ref{Theorem2}, if $\int_0^{\infty}e^{-\rho t}u(c^*(t),k^*(t))dt$ is defined in the sense of the Lebesgue integral, then $(k^*(t),c^*(t))\in A_{\bar{k}}$, and thus it is a solution to (\ref{MODEL}). Because $\inf_{t\ge 0}k^*(t)>0$, if $\inf_{t\ge 0}c^*(t)>0$, then $(k^*(t),c^*(t))\in A_{\bar{k}}$. However, this is not so easily verified. The following corollary presents three appropriate sufficient conditions for $(k^*(t),c^*(t))$ to be a solution to (\ref{MODEL}).

\begin{cor}
Suppose that Assumptions 1-6 hold, $k^*(t)$ is a solution to $(\ref{SOL2})$ defined on $\mathbb{R}_+$, and $c^*(t)$ is defined by $(\ref{CONSUM2})$. Suppose that one of the following three conditions holds.
\begin{enumerate}[1)]
\item $u(c,k)$ is bounded from above or below.

\item $k^*(t)$ is bounded.

\item $\liminf_{k\to \infty}c^*(p,k)>0$ for all $p>0$.
\end{enumerate}
Then, $(k^*(t),c^*(t))$ is a solution to $(\ref{MODEL})$.
\end{cor}

\begin{proof}
By Proposition \ref{Prop4} and Theorem \ref{Theorem1}, $\bar{V}$ is a solution to the HJB equation, $(k^*(t),c^*(t))\in B_{\bar{k}}$, and
\[\lim_{T\to\infty}\int_0^Te^{-\rho t}u(c^*(t),k^*(t))dt=\bar{V}(\bar{k}).\]
Therefore, it suffices to show that $(k^*(t),c^*(t))\in A_{\bar{k}}$.

For 1), if $u(c,k)$ is either bounded from above or below, then $A_{\bar{k}}=B_{\bar{k}}$, and thus our claim holds.

For 2), suppose that $k^*(t)$ is bounded. As we proved in Proposition \ref{Prop4}, $\inf_{t\ge 0}k^*(t)>0$. Therefore, the trajectory of $k^*(t)$ is included in some compact set $C\subset \mathbb{R}_{++}$. This implies that $\inf_{t\ge 0}c^*(t)>0$, and thus $e^{-\rho t}u(c^*(t),k^*(t))$ is bounded from below. Hence, $(k^*(t),c^*(t))\in A_{\bar{k}}$, as desired.

For 3), let $0<p_1<p_2$, and $c_i=c^*(p_i,k)$. Then,
\[F(k,c_i)p_i+u(c_i,k)=\max_{c\ge 0}\{F(k,c)p_i+u(c,k)\}.\]
By the first-order condition, there exists $r_i\in \partial_cF(k,c_i)$ such that
\[r_ip_i=-\frac{\partial u}{\partial c}(c_i,k).\]
If $c_1<c_2$, then
\[r_1\ge D_{c,+}F(k,c_1)\ge D_{c,-}F(k,c_2)\ge r_2,\]
and thus,\footnote{Note that $F$ is decreasing in $c$, and thus $r_i<0$.}
\[\frac{\partial u}{\partial c}(c_2,k)=-r_2p_2>-r_2p_1\ge -r_1p_1=\frac{\partial u}{\partial c}(c_1,k),\]
which contradicts the concavity of $u$. Therefore, $c^*(p,k)$ is nonincreasing in $p$. Define $\varepsilon=\inf_{t\ge 0}k^*(t)$. Then,
\[c^*(t)\ge c^*(D_-\bar{V}(k^*(t)),k^*(t))\ge \inf_{k\ge \varepsilon}c^*(D_-\bar{V}(\varepsilon),k)>0,\]
and thus $u(c^*(t),k^*(t))$ is bounded from below. Hence, again $(k^*(t),c^*(t))\in A_{\bar{k}}$, as desired. This completes the proof.
\end{proof}

Note that, the requirement in Corollary 1 is not strong. For example, suppose that $u(c,k)=au_{\theta}(c)+bu_{\theta}(k)$ for $a,b\ge 0$ for a CRRA function $u_{\theta}$ with $\theta\neq 1$, then $u(c,k)$ is either bounded from above or below, and 1) holds. Next, suppose that $F(k,c)=f(k)-dk-c$, $d>0$, $f$ is continuously differentiable, and $\lim_{k\to \infty}f'(k)<d$. Then, $k^+(t,\bar{k})$ is bounded, and thus $k^*(t)$ is also bounded, and 2) holds. Third, if $F(k,c)=f(k)-h(c)$ for a convex function $h$ and $u(c,k)=u(c)$, then $c^*(p,k)$ is independent of $k$. In this case, $\liminf_{k\to \infty}c^*(p,k)>0$ is trivially satisfied, and 3) holds.

The converse of Corollary 1 holds unconditionally. We state this result as a proposition.

\begin{prop}\label{Prop5}
Suppose that Assumptions 1-5 hold, and that there exists a solution $(k^*(t),c^*(t))$ to $(\ref{MODEL})$ such that $\inf_{t\ge 0}k^*(t)>0$. Then, $k^*(t)$ is a solution to $(\ref{SOL2})$, and $c^*(t)\in c^*(\partial \bar{V}(k^*(t)),k^*(t))$ for almost all $t\ge 0$.
\end{prop}

\begin{proof}
Suppose that $k^*(t)$ violates (\ref{SOL2}). Then, there exists $T>0$ such that the set $I=\{t\in [0,T]|c^*(t)\notin c^*(\partial \bar{V}(k^*(t)),k^*(t))\}$ is not a null set. By Theorem \ref{Theorem1}, $\bar{V}$ is a viscosity solution to the HJB equation. Define $W(t)=\bar{V}(k^*(t))$. By Lemma \ref{Lemma8}, $W(t)$ is absolutely continuous on $[0,T]$ and $\dot{W}(t)=D_-\bar{V}(k^*(t))\dot{k}^*(t)$ almost everywhere. Therefore,
\begin{align}
\int_0^Te^{-\rho t}u(c^*(t),k^*(t))dt=&~\int_0^Te^{-\rho t}[F(k^*(t),c^*(t))D_-\bar{V}(k^*(t))+u(c^*(t),k^*(t))]dt\nonumber \\
&~-\int_0^Te^{-\rho t}D_-\bar{V}(k^*(t))\dot{k}^*(t)dt\nonumber \\
<&~-\int_0^T[-\rho e^{-\rho t}W(t)+e^{-\rho t}\dot{W}(t)]dt\label{EVALL}\\
=&~\int_0^T\frac{d}{dt}[-e^{-\rho t}W(t)]dt\nonumber \\
=&~\bar{V}(\bar{k})-e^{-\rho T}\bar{V}(k^*(T)).\nonumber
\end{align}
By the same arguments, we have that for every $T'\ge T$,
\[\int_T^{T'}e^{-\rho t}u(c^*(t),k^*(t))dt\le e^{-\rho T}\bar{V}(k^*(T))-e^{-\rho T'}(k^*(T')).\]
By (\ref{EVALL}), there exists $\varepsilon>0$ such that
\[\int_0^Te^{-\rho t}u(c^*(t),k^*(t))dt<\bar{V}(\bar{k})-e^{-\rho T}\bar{V}(k^*(T))-2\varepsilon.\]
Choose $\delta>0$ such that $\inf_{t\ge 0}k^*(t)\ge \delta$. Then, for every $T'\ge T$ such that\footnote{Note that, because of Theorem 1, $\bar{V}$ satisfies (\ref{GC}).}
\[\max\{|e^{-\rho T'}\bar{V}(k^+(T',\bar{k}))|,|e^{-\rho T'}\bar{V}(\delta)|\}<\varepsilon,\]
we have that
\[\int_0^{T'}e^{-\rho t}u(c^*(t),k^*(t))dt<\bar{V}(\bar{k})-e^{-\rho T'}\bar{V}(k^*(T'))-2\varepsilon<\bar{V}(\bar{k})-\varepsilon.\]
Because $T'$ is arbitrary,
\[\int_0^{\infty}e^{-\rho t}u(c^*(t),k^*(t))dt\le \bar{V}(\bar{k})-\varepsilon<\bar{V}(\bar{k}),\]
which is a contradiction. Therefore, $k^*(t)$ is a solution to (\ref{SOL2}). Define $c^+(t)$ by (\ref{CONSUM2}). Then,
\[\dot{k}(t)=F(k^*(t),c^+(t))\]
for almost all $t\ge 0$, and
\[\dot{k}(t)=F(k^*(t),c^*(t))\]
for almost all $t\ge 0$, which implies that $c^*(t)=c^+(t)$ almost everywhere, and thus $c^*(t)\in c^*(\partial \bar{V}(k^*(t)),k^*(t))$ almost everywhere. This completes the proof.
\end{proof}

Finally, we discuss the differentiability of $\bar{V}(k)$. 

\begin{prop}\label{Prop6}
Suppose that Assumptions 1-4 hold and $\bar{V}$ is finite. Suppose also that if $F(k,c)=0$, then $F$ is differentiable in $c$ at $(k,c)$. Then, $\bar{V}$ is continuously differentiable, and thus it is a classical solution to the HJB equation.
\end{prop}

\begin{proof}
Because of Proposition \ref{Prop3}, $\bar{V}$ is an increasing and concave viscosity solution to the HJB equation. Moreover, because $\partial \bar{V}$ is upper hemi-continuous, if it is single-valued, then it is continuous. Hence, it suffices to show that $\bar{V}$ is differentiable everywhere. Suppose that $\bar{V}$ is not differentiable at $k$. Then, $\partial \bar{V}(k)=[D_+\bar{V}(k),D_-\bar{V}(k)]$, where $D_+\bar{V}(k)<D_-\bar{V}(k)$. Choose $p\in ]D_+\bar{V}(k),D_-\bar{V}(k)[$. By Lemma \ref{Lemma7},
\[F(k,c^*(p,k))p+u(c^*(p,k),k)=\rho \bar{V}(k),\]
\[F(k,c^*(p,k))q+u(c^*(p,k),k)\le \rho \bar{V}(k),\]
for all $q\in \partial \bar{V}(k)$. Therefore, $F(k,c^*(p,k))=0$. By the continuity of $F$ and $c^*$, we have that $F(k,c^*(p,k))=0$ for all $p\in \partial \bar{V}(k)$, and thus there exists $\bar{c}>0$ such that, $c^*(p,k)=\bar{c}$ for all $p\in \partial \bar{V}(k)$. By assumption, $F$ is differentiable in $c$ at $(k,\bar{c})$, and by the first-order condition,
\[\frac{\partial F}{\partial c}(k,\bar{c})p+\frac{\partial u}{\partial c}(\bar{c},k)=0\]
for all $p\in \partial \bar{V}(k)$, which is impossible because $\frac{\partial F}{\partial c}(k,\bar{c})<0$. This completes the proof.
\end{proof}

As a corollary, we obtain the following result.

\begin{cor}
Suppose that Assumptions 1-6 hold, and that if $F(k,c)=0$, then $F$ is continuously differentiable in $c$ at $(k,c)$. Then, $\bar{V}$ is a classical solution to the HJB equation in $\mathscr{V}$, and there is no other viscosity solution to the HJB equation in $\mathscr{V}$. Moreover, the following ODE
\begin{equation}\label{SOL3}
\dot{k}(t)=F(k(t),c^*(\bar{V}'(k(t)),k(t))),\ k(0)=\bar{k}
\end{equation}
has a solution $k^*(t)$ defined on $\mathbb{R}_+$ such that $\inf_{t\ge 0}k^*(t)>0$, and if we define $c^*(t)=c^*(\bar{V}'(k^*(t)),k^*(t))$, then $(k^*(t),c^*(t))\in B_{\bar{k}}$ and
\[\bar{V}(\bar{k})=\lim_{T\to \infty}\int_0^Te^{-\rho t}u(c^*(t),k^*(t))dt.\]
Furthermore, if one of the three conditions in Corollary 1 holds, then $(k^*(t),c^*(t))$ is a solution to $(\ref{MODEL})$ such that $k^*(t)$ is continuously differentiable and $c^*(t)$ is continuous.
\end{cor}

\begin{proof}
By Proposition \ref{Prop6} and Theorem \ref{Theorem2}, $\bar{V}$ is the unique classical solution to the HJB equation in $\mathscr{V}$. Therefore, $\partial \bar{V}(k)=\{\bar{V}'(k)\}$, and thus (\ref{SOL3}) is equivalent to (\ref{SOL2}). Hence, this has a solution $k^*(t)$ defined on $\mathbb{R}_+$ such that $\inf_{t\ge 0}k^*(t)>0$. Because the right-hand side of (\ref{SOL3}) is continuous in $k$, we have that $k^*(t)$ is continuously differentiable. Now, $c^*(t)$ is the same as that defined by (\ref{CONSUM2}). Therefore, $(k^*(t),c^*(t))\in B_{\bar{k}}$ and
\[\bar{V}(\bar{k})=\lim_{T\to \infty}\int_0^Te^{-\rho t}u(c^*(t),k^*(t))dt.\]
If one of the three conditions in Corollary 1 holds, then $(k^*(t),c^*(t))\in A_{\bar{k}}$, and thus it is a solution to (\ref{MODEL}). By definition, $c^*(t)$ is continuous. This completes the proof.
\end{proof}

\section{Examples}
\subsection{Case of the RCK model}
Our model (\ref{MODEL}) is much more complex than the most commonly used class of models in economics. This makes it difficult to understand what our assumptions imply. In this section, we examine the assumptions of the traditional RCK model, and clarify the implications of Assumptions 1-6.

The model we discuss in this section is as follows.
\begin{align}
\max~~~~~&~\int_0^{\infty}e^{-\rho t}u(c(t))dt \nonumber \\
\mbox{subject to. }&~c(\cdot)\in W_1,\nonumber \\
&~k(t)\mbox{ is absolutely continuous on any compact interval},\nonumber \\
&~k(t)\ge 0,\ c(t)\ge 0,\label{MODEL2}\\
&~\int_0^{\infty}e^{-\rho t}u(c(t))dt\mbox{ can be defined},\nonumber \\
&~\dot{k}(t)=f(k(t))-dk(t)-c(t)\mbox{ a.e.},\nonumber \\
&~k(0)=\bar{k}.\nonumber
\end{align}

The background story of this model is as follows. As in the model (\ref{MODEL}), $k(t)$ denotes the amount of capital stock, and $c(t)$ denotes the amount of consumption. In addition, let $y(t)$ denote the total production and $i(t)$ denote the amount of investment. The function $f$ denotes the {\bf production function}. We assume the following three relationships. First, the total production $y(t)$ is determined by the formula
\[y(t)=f(k(t)).\]
Second, the following simplified {\bf IS relationship}
\[y(t)=c(t)+i(t)\]
is assumed. Third, the increasing speed of capital stock is determined by the following formula
\[\dot{k}(t)=i(t)-dk(t),\]
where $d$ is the capital depreciation rate. Combining these three equations, we obtain that
\[\dot{k}(t)=f(k(t))-dk(t)-c(t),\]
which appears in the model (\ref{MODEL2}).

We make the following assumptions.

\vspace{12pt}
\noindent
{\bf Assumption 1'}. $\rho>0$.

\vspace{12pt}
\noindent
{\bf Assumption 2'}. The instantaneous utility function $u:\mathbb{R}_+\to \mathbb{R}\cup\{-\infty\}$ is a continuous, concave, and increasing function. Moreover, $u$ is continuously differentiable on $\mathbb{R}_{++}$, $u'(c)$ is decreasing, and $\lim_{c\to 0}u'(c)=+\infty$, $\lim_{c\to \infty}u'(c)=0$.

\vspace{12pt}
\noindent
{\bf Assumption 3'}. The production function $f:\mathbb{R}_+\to \mathbb{R}_+$ is a continuous, concave, and increasing function that satisfies $f(0)=0$. Moreover, $d\ge 0$ and there exists $k>0$ such that $f(k)>dk$.

\begin{prop}\label{Prop7}
Suppose that Assumptions 1'-3' holds. Define $u(c,k)=u(c)$ and $F(k,c)=f(k)-dk-c$. Then, Assumptions 1-4 holds, and $F$ is continuously differentiable in $c$. In particular, the value function $\bar{V}$ is a nondecreasing and concave function such that $\bar{V}(k)>-\infty$ for all $k>0$, and if $\bar{V}$ is finite, then it is an increasing function that is a classical solution to the HJB equation.
\end{prop}

\begin{proof}
First, because $W=W_1$, Assumption 1 holds. Second, Assumption 2 requires that $u$ is a continuous and concave function on $\mathbb{R}_+$ that is increasing and continuously differentiable on $\mathbb{R}_{++}$. Assumption 4 requires that $u'$ is decreasing in $c$, $\lim_{c\to 0}u'(c)=+\infty$, and $\lim_{c\to \infty}u'(c)=0$. All requirements are satisfied under Assumption 2'. Third, Assumption 3 requires that $f$ is a continuous and concave function that satisfies $f(0)=0$, $d\ge 0$, and there exists $k>0$ such that $f(k)>dk$. All requirements are satisfied under Assumption 3'. Moreover, for our $F(k,c)$, $\frac{\partial F}{\partial c}(k,c)\equiv -1$. By Proposition \ref{Prop2}, $\bar{V}$ is a nondecreasing and concave function such that $\bar{V}(k)>-\infty$ for all $k>0$. Moreover, by Propositions \ref{Prop3} and \ref{Prop6}, if $\bar{V}$ is finite, then it is an increasing function that is a classical solution to the HJB equation, as desired. This completes the proof.
\end{proof}

We say that $f$ satisfies the {\bf Inada condition} if $f$ is continuously differentiable, strictly concave, and $f'(\mathbb{R}_{++})=\mathbb{R}_{++}$. Then, the following proposition holds.

\begin{prop}\label{Prop8}
Suppose that Assumptions 1'-3' hold, and that $f$ satisfies the Inada condition. Define $u(c,k)=u(c)$ and $F(k,c)=f(k)-dk-c$. Then, Assumption 6 holds. Moreover, if there exist $a>0, C\in \mathbb{R}$, and $\theta>0$ such that
\begin{equation}\label{BOUND}
u(c)\le au_{\theta}(c)+C
\end{equation}
for all $c\ge 0$, then Assumption 5 holds, and thus, $\bar{V}$ is an increasing and concave classical solution to the HJB equation, and there is no other increasing and concave viscosity solution to the HJB equation. Furthermore, a solution $(k^*(t),c^*(t))$ to $(\ref{MODEL2})$ can be obtained from $(\ref{SOL2})$ and $(\ref{CONSUM2})$, and $k^*(t)$ is continuously differentiable and $c^*(t)$ is continuous.
\end{prop}

\begin{proof}
First, $\frac{\partial^2u}{\partial k\partial c}(c,k)\equiv 0$. By the Inada condition, there exists $k>0$ such that $f'(k)>\rho$. Because $\frac{\partial^2F}{\partial k\partial c}(k,c)\equiv 0$ for all $(k,c)$, we can choose $\varepsilon_0=1,\ H(k,c)\equiv k-c$. Therefore, Assumption 6 holds.\footnote{Note that, we do not prohibit the differentiability of $F$ at $(k,c)$ with $H(k,c)=0$. For Assumption 6, it is only necessary that if $k<\varepsilon_0$ and $H(k,c)\neq 0$, then $F$ satisfies some differentiability requirements, and the converse is not needed.} Second, suppose that there exist $a>0, C\in \mathbb{R}$, and $\theta>0$ such that (\ref{BOUND}) holds. By the Inada condition, there exists $k>0$ such that $f'(k)>d$ and $\rho-(1-\theta)(f'(k)-d)>0$. Clearly, $(f'(k)-d,-1)\in \partial F(k,1)$. If we choose $b=0$, then all the requirements in Assumption 5 holds. Therefore, Assumptions 1-6 holds, and by Theorem \ref{Theorem2}, $\bar{V}$ is the unique viscosity solution to the HJB equation in $\mathscr{V}$. Note that, by Lemma \ref{Lemma2}, $\mathscr{V}$ is the set of all increasing and concave functions, and thus there is no increasing and concave viscosity solution to the HJB equation other than $\bar{V}$. By Proposition \ref{Prop7}, $\bar{V}$ is a classical solution to the HJB equation. Finally, because $c^*(p,k)=(u')^{-1}(p)$, statement 3) of Corollary 1 holds, and thus Corollary 2 can be applied. Hence, a solution $(k^*(t),c^*(t))$ to (\ref{MODEL2}) is obtained by (\ref{SOL3}) and the equality $c^*(t)=f(k^*(t))-dk^*(t)-\dot{k}^*(t)$, and $k^*(t)$ is continuously differentiable and $c^*(t)$ is continuous. This completes the proof.
\end{proof}

Hence, in the usual RCK model with the Inada condition, (\ref{BOUND}) is the unique requirement for ensuring that the value function is the unique classical solution to the HJB equation.

\subsection{The Non-Smooth Fiscal Policy}
In this section, we treat an economic example of the model in which Assumptions 1-6 are satisfied, but $F(k,c)$ is not differentiable in $c$. In this model, $\bar{V}$ may not be differentiable, and thus it may not be a classical solution to the HJB equation, despite being a viscosity solution to the HJB equation.

The model is as follows:
\begin{align}
\max~~~~~&~\int_0^{\infty}e^{-\rho t}u(c(t))dt \nonumber \\
\mbox{subject to. }&~c(\cdot)\in W_1,\nonumber \\
&~k(t)\mbox{ is absolutely continuous on any compact interval},\nonumber \\
&~k(t)\ge 0,\ c(t)\ge 0,\label{MODEL3}\\
&~\int_0^{\infty}e^{-\rho t}u(c(t))dt\mbox{ can be defined},\nonumber \\
&~\dot{k}(t)=f(k(t))-dk(t)-c(t)-g(k(t),c(t))\mbox{ a.e.},\nonumber \\
&~k(0)=\bar{k},\nonumber
\end{align}
where $g(k,c)=\max\{[Af(k)-c]B,0\}$ and $0<A,B<1$.

Let us explain the background story of this model. As in the RCK model, $f$ denotes the production function and $d\ge 0$ is the capital depreciation rate. The difference from the RCK model is the presence of the {\bf fiscal policy rule} $g(k,c)$. We assume that the government conducts a Keynesian policy to attempt to buoy the economy when consumption is too small relative to production. In this model, the boundary value of consumption at which the government decides to spend is $100A$\% of the output, and government expenditure is determined as $100B$\% of the shortfall by which consumption is below that boundary value. The function $g(k,c)$ reflects this policy rule. Through this fiscal policy, the government produces some `public' commodity that increases the consumer's utility. However, in the decentralized economy behind this model, the action of the government is independently determined, and the consumer does not consider changing government spending through his/her own action. Thus, the consumer's behavioural decision does not reflect the utility derived from the `public' commodity produced by the government, and as a result, it is excluded from the model's objective function. Therefore, the instantaneous utility function is assumed to be a function of $c$ only.

Let $u(c,k)=u(c)$ and $F(k,c)=f(k)-c-g(k,c)$. Suppose that Assumptions 1'-3' hold, $f$ satisfies the Inada condition, and there exist $a>0$, $C\in \mathbb{R}$, and $\theta>0$ such that (\ref{BOUND}) holds. We check that model (\ref{MODEL3}) satisfies Assumptions 1-6. First, by the same arguments as in the proof of Propositions \ref{Prop7} and \ref{Prop8}, Assumptions 1, 2, 4, and 5 are clearly satisfied. For Assumption 3, consider
\[G(y,c)=y-c-\max\{(Ay-c)B,0\}.\]
It is easy to show that $G$ is concave and increasing in $y$. Then,
\begin{align*}
&~F((1-t)(k_1,c_1)+t(k_2,c_2))\\
=&~G(f((1-t)k_1+tk_2),(1-t)c_1+tc_2)-d[(1-t)k_1+tk_2]\\
\ge&~G((1-t)(f(k_1),c_1)+t(f(k_2),c_2))-d[(1-t)k_1+tk_2]\\
\ge&~(1-t)[G(f(k_1),c_1)-dk_1]+t[G(f(k_2),c_2)-dk_2]\\
=&~(1-t)F(k_1,c_1)+tF(k_2,c_2),
\end{align*}
and thus $F$ is concave. The rest of the claim in Assumption 3 clearly holds with $d_1=d, \delta_2(c)=c$. Finally, let $H(k,c)=Af(k)-c$. Then, we can easily check that Assumption 6 holds. Hence, we have that the value function $\bar{V}$ in (\ref{MODEL3}) is the unique viscosity solution to the HJB equation in $\mathscr{V}$. Because $F(k,0)=(1-AB)f(k)-dk$, Lemma \ref{Lemma2} implies that $\mathscr{V}$ is the set of all increasing and concave real-valued functions on $\mathbb{R}_{++}$.

We can easily show that
\[c^*(p,k)=\begin{cases}
(u')^{-1}(p) & \mbox{if }p\le u'(Af(k)),\\
Af(k) & \mbox{if }p\ge u'(Af(k))\ge (1-B)p,\\
(u')^{-1}((1-B)p) & \mbox{if }(1-B)p\ge u'(Af(k)),
\end{cases}\]
and thus, statement 3) of Corollary 1 holds. Therefore, a solution $(k^*(t),c^*(t))$ to (\ref{MODEL3}) can be obtained from (\ref{SOL3}) and the equality $c^*(t)=c^*(\bar{V}'(k^*(t)),k^*(t))$.

Surprisingly, if $d=0$, then $\bar{V}$ is continuously differentiable. Indeed, we can easily check that if $F(k,c^*(p,k))=0$, then $c^*(p,k)=f(k)>Af(k)$, and thus $F$ is differentiable at $(k,c^*(p,k))$. Hence, Proposition \ref{Prop6} implies that $\bar{V}$ is continuously differentiable.

However, if $d>0$, then there exists $k^*$ such that $f(k^*)-dk^*=Af(k^*)$. In this case, it is unknown whether $\bar{V}$ is differentiable at $k^*$.

\section{Discussion}
Usually, the HJB equation is written as a second-order degenerate elliptic differential equation. In fact, in stochastic economic dynamic models, the HJB equation becomes a second-order differential equation. For example, Malliaris and Brock \cite{MB} dealt with this type of equation. Probably, the use of Ito's formula in the middle of the derivation makes the HJB equation second-order. On the other hand, since there is no stochastic variation in the dynamic model considered in this paper, the HJB equation is only a first-order differential equation.

The case in which the HJB equation does not have a classical solution has been highlighted in many studies. Therefore, following Lions \cite{LI}, the viscosity solution is typically taken as the solution concept in the study of degenerate elliptic differential equations. However, as in Proposition \ref{Prop6}, the condition for the value function to be differentiable is very mild, and thus, in many economic models, the value function is a classical solution to the HJB equation. For such a model, the net contribution of this paper is to show the uniqueness of the viscosity, not classical, solution to the HJB equation rigorously.

As noted in the introduction, the properties of the HJB equation presented in many previous studies are not applicable to the model considered in this paper. There are several reasons for this. First, in the model considered herein, the time interval is infinite. To the best of our knowledge, this type of problem was first addressed by Soner \cite{SO1, SO2}, who made the assumption that moving the control variable of the cost function does not have a significant effect on the cost. Applying this assumption to economic models, we must assume that $u(c,k)$ is bounded in $c$. However, the most typical example of $u(c,k)$ in economics is $\log c$, which violates this assumption. Barles \cite{BA} pointed out that when this boundedness requirement does not hold, a pathological problem arises. Hosoya \cite{HOA, HOB} showed that this pathological problem appears even in RCK models, and thus this problem is serious for economic models.

This boundedness assumption is often applied in the modern theory of the HJB equation for control problems with infinite time interval. Bardi and Capuzzo-Dolcetta \cite{BCD} is a typical example. In contrast, several recent studies have not assumed the boundedness of $u$. For example, Hermosilla et al. \cite{HVZ} treated such a case for a problem with a finite time interval. However, even in that case, $u$ is assumed to be Lipschitz. This is a condition that rarely applies to economic models, because any CRRA function is not Lipschitz around $0$. Da Lio \cite{D} also treated an infinite time interval model with unbounded $u$. However, in this paper, $u$ is still bounded from above.\footnote{Da Lio \cite{D} treated the minimization problem, and thus he assumed the boundedness from below. However, we treat the maximization problem, and thus `below' must be changed into `above'. Also, this paper assumed that $F$ satisfies some Lipschitz-like condition, which is usually inconsistent with the Inada condition.} This condition is violated if $u=u_{\theta}$ with $\theta\le 1$, and thus, these requirements cannot be applied for many economic applications. To the best of our knowledge, no model other than that of Hosoya \cite{HOA, HOB} has dealt with the case of $u(c,k)=\log c$. Hosoya \cite{HOA} allows the differentiability of $F$ to derive the results, but this is not assumed in the present paper. Hosoya \cite{HOB} assumes the existence of a continuous solution to the original problem, but again, the present paper does not make such an assumption. Hence, the results presented in this paper are independent of those derived in previous studies.

On the other hand, in economics, if $c$ can be written as a function of $k$ and $\dot{k}$, the value function is differentiable when the original problem has a solution such that $\dot{k}(t)$ is continuous. This is because of Theorem 2 of Benveniste and Scheinkman \cite{BS2}. Thus, the value function $\bar{V}$ becomes non-differentiable only if the optimal capital accumulation path $k^*(t)$ is not continuously differentiable. In other words, the value function can be non-differentiable only when (\ref{SOL2}) is truly a differential inclusion. Unfortunately, (\ref{SOL2}) contains information on $\bar{V}$ itself, and thus deriving the non-differentiability of $\bar{V}$ from the assumptions of $u$ and $F$ involves considerable difficulty. However, there exist cases in which it is not possible to judge whether the value function is differentiable. In such cases, our Theorems 1-2 serve as powerful analytical tools.

In this paper, the fact that $u$ may be unbounded was a major problem. It is possible to remove this problem by restricting the range of values $c(t)$ can take. However, in that case, it is not possible to deal with problems in which $F(k,c)=Ak-c$. Such a model is called the AK model in economics. It is known that the solution to the AK model usually satisfies $\lim_{t\to \infty}c(t)=+\infty$, and thus it is not possible to deal with the AK model if an upper bound of $c(t)$ is introduced. It is easy to show that there exists an AK model that is consistent with all of Assumptions 1-6, and thus in this paper, by not setting an upper bound of $c(t)$, the range of models that can be handled increases.

We additionally mention the so-called ``inward pointing condition''. This condition states that if a state variable is at the boundary of the constraint set, then there exists a control variable that can bring this state variable back to the interior of the constraint set. See, for example, (1.3) of Colombo et al. \cite{CKR}. In our model (\ref{MODEL}), this condition states that there exists $c\ge 0$ such that $F(0,c)>0$. However, this condition contradicts our Assumption 3, because $F(0,0)=0$ and $F$ is decreasing in $c$. Colombo et al. \cite{CKR} mentioned a sort of `higher-order' inward pointing conditions in which the partial derivative of $F$ in $k$ at $(0,0)$ is used. However, the Inada condition implies that $F$ is not differentiable in $k$ at $(0,0)$, and thus a differentiable requirement of $F$ at the boundary does not hold in many economic models.

This ``inward pointing condition'' is used to guarantee a certain kind of viability of the state variable function, and thus, it is similar in spirit to Assumption 6, which guarantees the existence of a positive solution to the differential inclusion (\ref{SOL}) defined on $\mathbb{R}_+$. However, Assumption 6 is a condition that guarantees that ``for a candidate of the solution to (\ref{MODEL}), the state variable will never reach the boundary,'' whereas the inward pointing condition is a sort of conditions that guarantee that ``when a candidate of the solution reaches the boundary, it can return to the interior.'' Hence, the two conditions are different in nature. In the usual cases, $u$ and $F$ in the economic models do not have good properties at the boundary, and in this sense, any condition that guarantees a good property of the model on the boundary is not acceptable in economic models. Hence, it is unlikely that such a `well-behaved' condition on the boundary will be of much use in economic models in the future.

\section{Concluding Remarks}
In this paper, we considered a class of economic control problems, and presented the conditions under which the value function is the unique viscosity solution to the HJB equation. For such cases, we provided a method for deriving a solution to the original problem using the HJB equation. Furthermore, we presented a condition under which the value function is differentiable. Our conditions are sufficiently weak that our results can be applied to many economic applications.

There are several future tasks. First, we want to prove that Corollary 1 holds unconditionally.

Second, we want to extend our results to some multidimensional models and stochastic models.

Third, we want to obtain a simple method for gaining an approximate solution to the value function. In discrete-time models, there is a famous approximation method that uses Blackwell's inequality and the contraction mapping theorem.\footnote{See Stokey and Lucas \cite{SL} for more detailed arguments.} We want to obtain a counterpart to this result for a continuous-time model.

Fourth, we want to extend our result to decentralized models in economic theory. Usually, macroeconomic dynamic models can be classified into two categories: centralized and decentralized. We have only discussed centralized models in this paper. If the government is absent, then by the fundamental theorem of welfare economics, the results of the two models coincide. However, if fiscal policy is introduced, then these two models can derive different results. Hence, we want to extend our result to decentralized models.

\section*{Acknowledgements}
The author is grateful to anonymous reviewers for their kind comments and suggestions. Moreover, the author thanks Toru Maruyama and Hisatoshi Tanaka for their helpful suggestions in private discussions. This work was supported by JSPS KAKENHI Grant Number JP25K05007.

\section*{Declaration}
The author reports there are no competing interests to declare.

\appendix
\section{Proofs of Lemmas}
In this appendix, we put on the proofs of Lemmas 5-8.

\begin{proof}[{\bf Proof of Lemma \ref{Lemma5}}]
Suppose not. Then, there exists $t^+>0$ such that $k_1(t^+)<k_2(t^+)$. Let $t^*=\inf\{t\in [0,T]|\forall s\in [t,t^+],\ k_1(s)<k_2(s)\}$. Because $k_1(0)=k_2(0)=\bar{k}$, we have that $k_1(t^*)=k_2(t^*)$. Let
\[k_3(t)=k_1(t^*)+\int_{t^*}^th(k_2(s))ds.\]
Because $\dot{k}_2(t)\in \Gamma(k_2(t))$, we have that $h(k_2(t))\ge \dot{k}_2(t)$ almost everywhere, and thus, 
\[k_3(t)\ge k_2(t)>k_1(t)\]
for all $t\in ]t^*,t^+]$. Let $L>0$ be the Lipschitz constant on a compact set $C$ that includes $\{k_i(t)|i\in \{1,2\},\ t\in [t^*,t^+]\}$. Then,
\begin{align*}
k_3(t)-k_1(t)=&~\int_{t^*}^t[h(k_2(s))-h(k_1(s))]ds\\
\le&~\int_{t^*}^tL[k_2(s)-k_1(s)]ds\\
\le&~L(t-t^*)\max_{s\in [t^*,t]}(k_2(s)-k_1(s)).
\end{align*}
Choose any $t\in [t^*,t^+]$ such that $0<t-t^*<L^{-1}$, and let $s^*\in \arg\max\{k_2(s)-k_1(s)|s\in [t^*,t]\}$. Because $k_1(t^*)=k_2(t^*)$ and $k_2(s)>k_1(s)$ for $s\in ]t^*,t]$, we have that $s^*>t^*$. Then,
\[k_3(s^*)-k_1(s^*)\le L(s^*-t^*)(k_2(s^*)-k_1(s^*))<k_2(s^*)-k_1(s^*),\]
and thus, $k_3(s^*)<k_2(s^*)$, which is a contradiction. This completes the proof.
\end{proof}

\begin{proof}[{\bf Proof of Lemma \ref{Lemma6}}]
Let $Y$ be the set of all solutions to (\ref{DI}) defined on either $\mathbb{R}_+$ or $[0,T']$ for some $T'$. Because (\ref{DI}) has a solution, $Y$ is nonempty. For $k_1(\cdot),k_2(\cdot)\in Y$, let $k_1(\cdot)\succeq k_2(\cdot)$ if the domain of $k_1(\cdot)$ includes that of $k_2(\cdot)$ and $k_1(t)=k_2(t)$ when both are defined at $t$. Clearly, $\succeq$ is a partial order on $Y$. For $k(\cdot)\in Y$, let $I_{k(\cdot)}$ be the domain of $k(\cdot)$. Choose any chain $C\subset Y$ of $\succeq$. If $\sup \cup_{k(\cdot)\in C}I_{k(\cdot)}=+\infty$, then we can define
\[k^+(t)=k(t)\mbox{ if }t\in I_{k(\cdot)},\]
and $k^+(\cdot)$ is an upper bound of $C$. Otherwise, let $T^*=\sup \cup_{k(\cdot)\in C}I_{k(\cdot)}$. Define
\[k^+(t)=k(t)\mbox{ if }t\in I_{k(\cdot)}.\]
Then, $k^+(t)$ is a solution to (\ref{DI}) defined on $[0,T^*[$. By the continuity of $\hat{k}(t),\bar{k}(t)$, there exist $\varepsilon>0$ and $\delta>0$ such that $\bar{k}(T^*)>\varepsilon$ and if $0<T^*-t<\delta$, then $k^+(t)\in [\bar{k}(T^*)-\varepsilon, \hat{k}(T^*)+\varepsilon]$. Hence, we can define
\[k^+(T^*)=\limsup_{t\to T^*}k^+(t)\in [\bar{k}(T^*)-\varepsilon,\hat{k}(T^*)+\varepsilon].\]
Because $[\bar{k}(T^*)-\varepsilon, \hat{k}(T^*)+\varepsilon]$ is compact and $\Gamma$ is upper hemi-continuous, there exists $r>0$ such that $|y|<r$ for every $k\in [\bar{k}(T^*)-\varepsilon,\hat{k}(T^*)+\varepsilon]$ and $y\in \Gamma(k)$. Thus, $|k^+(t)-k^+(T^*)|\le r(T^*-t)$, and hence
\[\lim_{t\to T^*}k^+(t)=k^+(T^*),\]
which implies that $k^+(t)$ is continuous on $[0,T^*]$. Then, $k^+(t)$ is differentiable almost everywhere on $[0,T^*]$ and 
\[k^+(t)=\bar{k}+\int_0^t\dot{k}^+(t)ds\]
for all $t\in [0,T^*]$, and thus, we conclude that $k^+(t)$ is a solution to equation (\ref{DI}) defined on $[0,T^*]$. Hence, $k^+(\cdot)\in Y$ and it is an upper bound of $C$. By Zorn's lemma, there exists a maximal element $k^*(\cdot)\in Y$ of $\succeq$. If the domain $I_{k^*(\cdot)}$ is $[0,T]$, then $k^*(T)\in [\bar{k}(T),\hat{k}(T)]$, and thus as mentioned above, there exists a solution $k_1(t)$ to the following inclusion:
\[\dot{k}(t)\in \Gamma(k(t)),\ k(0)=k^*(T),\]
defined on $[0,T']$. Define
\[\kappa(t)=\begin{cases}
k^*(t) & \mbox{if }t\in [0,T],\\
k_1(t-T) & \mbox{if }t\in [T,T+T'].
\end{cases}\]
Then, $\kappa(\cdot)\in Y$ and $\kappa(\cdot)\succ k^*(\cdot)$, which is a contradiction. Therefore, the domain of $k^*(t)$ must be $\mathbb{R}_+$. This completes the proof.
\end{proof}

\begin{proof}[{\bf Proof of Lemma \ref{Lemma7}}]
First, suppose that $V$ is a viscosity solution to the HJB equation. If $V$ is differentiable at $k$, then $p=V'(k)$, and thus (\ref{EQQ}) holds. Suppose that $\partial V(k)=[D_+V(k), D_-V(k)]$ and $D_+V(k)<D_-V(k)$. Because $\partial V$ is upper hemi-continuous and $V$ is locally Lipschitz, there exist sequences $(k_m^1),(k_m^2)$ such that $V$ is differentiable at $k_m^i$, and
\[\lim_{m\to \infty}V'(k_m^1)=D_+V(k),\ \lim_{m\to \infty}V'(k_m^2)=D_-V(k).\]
Then,
\[\rho V(k_m^i)=F(k,c^*(V'(k_m^i),k_m^i))V'(k_m^i)+u(c^*(V'(k_m^i),k_m^i),k_m^i),\]
and letting $m\to \infty$, we have that (\ref{EQQ}) holds for $p=D_+V(k)$ or $p=D_-V(k)$. Now, define
\[g(p)=\sup_{c\ge 0}\{F(k,c)p+u(c,k)\}.\]
Then, we can easily check that $g$ is convex. Because $g(D_+V(k))=g(D_-V(k))=\rho V(k)$, we have that $g(p)\le \rho V(k)$ for all $p\in \partial V(k)$. Now, choose any $p\in \partial V(k)$, and define $\varphi(k')=V(k)+p(k'-k)$. Then, $\varphi(k)=V(k)$ and $\varphi'(k')\ge V(k')$ for all $k'$. Because $V$ is a viscosity supersolution, we have that $g(p)=g(\varphi'(k))\ge \rho V(k)$, and thus $g(p)\equiv \rho V(k)$ on $\partial V(k)$, as desired.

Conversely, suppose that (\ref{EQQ}) holds for any $p\in \partial V(k)$. Choose any continuously differentiable function $\varphi$ defined on some neighbourhood of $k$ such that $k$ is either a minimum point or a maximum point of $\varphi-V$. Then, $\varphi'(k)\in \partial V(k)$, and thus (\ref{EQQ}) holds for $p=\varphi'(k)$, which implies that $V$ is a viscosity solution to the HJB equation. This completes the proof.
\end{proof}

\begin{proof}[{\bf Proof of Lemma \ref{Lemma8}}]
First, let $C=x([A,B])$. Then, $C$ is a compact and convex subset of $U$, and thus $H$ has a Lipschitz constant $L>0$ on $C$. Choose any $\varepsilon>0$. Then, there exists $\delta>0$ such that if $s_1<t_1\le s_2<t_2\le ... \le s_n<t_n$ and $\sum_{i=1}^n|t_i-s_i|<\delta$, then $\sum_{i=1}^n|x(t_i)-x(s_i)|<L^{-1}\varepsilon$, which implies that
\[\sum_{i=1}^n|\psi(t_i)-\psi(s_i)|\le L\sum_{i=1}^n|x(t_i)-x(s_i)|<\varepsilon,\]
as desired. Therefore, $\psi$ is absolutely continuous.

Second, suppose that $\dot{\psi}(t)$ and $\dot{x}(t)$ are defined, and $p\in \partial H(x(t))$. Then,
\begin{align*}
\dot{\psi}(t)-p\dot{x}(t)=&~\lim_{h\downarrow 0}\frac{\psi(t+h)-\psi(t)}{h}-p\dot{x}(t)\\
\le&~\lim_{h\downarrow 0}p\frac{x(t+h)-x(t)}{h}-p\dot{x}(t)=0,\\
\dot{\psi}(t)-p\dot{x}(t)=&~\lim_{h\uparrow 0}\frac{\psi(t+h)-\psi(t)}{h}-p\dot{x}(t)\\
\ge&~\lim_{h\uparrow 0}p\frac{x(t+h)-x(t)}{h}-p\dot{x}(t)=0,
\end{align*}
which implies that $\dot{\psi}(t)=p\dot{x}(t)$. This completes the proof.
\end{proof}

\end{document}